\documentclass[10pt]{article}

\usepackage[utf8]{inputenc} 
\usepackage[T1]{fontenc}    
\usepackage{hyperref}       
\usepackage{url}            
\usepackage{booktabs}       
\usepackage{amsfonts}       
\usepackage{nicefrac}       

\usepackage{amsmath,amssymb,amsthm,nccmath,mathtools,bbm}

\usepackage[algo2e,ruled,noend]{algorithm2e}
\usepackage{enumitem}
\usepackage{multirow}
\usepackage{array}
\usepackage{graphicx,subcaption}

\usepackage{geometry,comment}
\usepackage[noblocks]{authblk}

\usepackage{xcolor}
\usepackage{pbox,adjustbox}

\newcommand{\E}{\mathbb{E}}
\renewcommand{\vec}{\mathrm{vec}}

\newcommand{\mK}{\mathsf K}
\newcommand{\one}{\mathbbm{1}}

\newcommand{\Pb}{\mathbb P}
\newcommand{\Sb}{\mathbb S}
\newcommand{\R}{\mathbb R}
\newcommand{\mD}{\mathsf D}
\newcommand{\mG}{\mathsf G}

\newcommand{\Q}{\mathcal Q}

\newcommand{\nbf}{\noindent\textbf}
\newcommand{\nit}{\noindent\textit}

\newtheorem{theorem}{Theorem}
\newtheorem{proposition}{Proposition}
\newtheorem{lemma}{Lemma}

\newtheorem{assumption}{Assumption}

\theoremstyle{definition}

\theoremstyle{remark}
\newtheorem{remark}{Remark}

\makeatletter
\newcommand{\removelatexerror}{\let\@latex@error\@gobble}
\makeatother

\setlist[itemize]{leftmargin=*}

\title{\bf Distributed Reinforcement Learning for Decentralized Linear Quadratic Control: A Derivative-Free Policy Optimization Approach%
}
\date{}

\author{Yingying Li, Yujie Tang, Runyu Zhang, and Na Li\thanks{The authors are with the John A. Paulson School of Engineering and Applied Sciences, Harvard University. Emails: \texttt{yingyingli@g.harvard.edu}, \texttt{yujietang@seas.harvard.edu}, \texttt{runyuzhang@fas.harvard.edu}, \texttt{nali@seas.harvard.edu}.
\newline\indent The first two authors contribute equally.}%
}

\begin{document}

\maketitle

\begin{abstract}
This paper considers a distributed reinforcement learning problem for decentralized linear quadratic control with partial state observations and local costs. We propose a Zero-Order Distributed Policy Optimization algorithm (ZODPO) that learns linear local controllers in a distributed fashion, leveraging the ideas of policy gradient, zero-order optimization and consensus algorithms. In ZODPO, each agent estimates the global cost by consensus, and then conducts local policy gradient in parallel based on  zero-order gradient estimation. ZODPO only requires limited communication and storage even in large-scale systems. Further, we investigate the nonasymptotic performance of ZODPO and show that the sample complexity to approach a stationary point is polynomial with the  error tolerance's inverse and the problem dimensions, demonstrating the scalability of ZODPO. We also show that the controllers generated throughout ZODPO are stabilizing controllers with high probability. Lastly, we numerically test ZODPO on  multi-zone HVAC systems.

\vspace{5pt}
\end{abstract}

\section{Introduction}
Reinforcement learning (RL) has emerged as a promising tool for controller design for dynamical systems, especially when the system model is unknown or complex, and has wide applications in, e.g., robotics \cite{riedmiller2009reinforcement}, games \cite{silver2017mastering},
 manufacturing \cite{wang2005application}, autonomous driving \cite{shah2018airsim}.
However, theoretical performance guarantees of RL are still under-developed across a wide range of problems, limiting the application of RL to real-world systems. Recently, there have been exciting theoretical results on learning-based control for (centralized) linear quadratic (LQ) control problems \cite{dean2017sample,fazel2018global,ouyang2017learning}.
LQ control is one of the most well-studied optimal control problems, which considers optimal state feedback control for a linear dynamical system such that a quadratic cost on the states and control inputs is minimized over a finite or infinite horizon
\cite{lewis2012optimal}.

Encouraged by the recent success of learning-based centralized LQ control, this paper aims to extend the results and develop scalable learning algorithms for decentralized LQ control. In decentralized control, the global system is controlled by a group of individual agents with limited communication, each of which observes only a partial state of the global system \cite{bakule2008decentralized}. Decentralized LQ control has many applications, including transportation \cite{bazzan2009opportunities}, power grids \cite{pipattanasomporn2009multi}, robotics \cite{cao1997cooperative}, smart buildings \cite{zhang2016decentralized}, etc. It is worth mentioning that partial observations and limited communication place major challenges on finding optimal decentralized controllers, even when the global system model is known \cite{witsenhausen1968counterexample,rotkowitz2005characterization}. 

{Specifically, we consider the following decentralized LQ control setting.  Suppose a linear dynamical system, with a global state $x(t)\in\mathbb{R}^n$ and a global control action $u(t)$, is controlled by a group of agents. The global control action is composed of local control actions: $
u(t)=\left[u_1(t)^\top, \ldots,u_N(t)^\top\right]^\top$, where $u_i(t)$ is the control input of agent $i$. At time $t$, each agent $i$ directly observes a partial state $x_{\mathcal{I}_i}(t)$ and a quadratic local cost $c_i(t)$ that could depend on the global state and action. The dynamical system model is assumed to be unknown, and the agents can only communicate with their neighbors via a communication network.  The goal is to design a cooperative distributed learning scheme to find local control policies for the agents to minimize the global cost that is averaged both among all agents and across an infinite horizon. The local control policies are limited to those that only use local observations.}


\subsection{Our contributions}

{We propose a Zero-Order Distributed Policy Optimization algorithm (ZODPO) for the decentralized LQ control problem defined above. ZODPO only requires limited communication over the network and limited storage of local policies, thus being applicable for large-scale systems.
Roughly, in ZODPO, each agent updates its local control policy using estimate of the partial gradient of the global objective with respect to its local policy. The partial gradient estimation leverages zero-order optimization techniques, which only requires cost estimation of a perturbed version of the current policies. To ensure distributed learning/estimation, we design an approximate sampling method to generate policy perturbations under limited communication among agents; we also develop a consensus-based algorithm to estimate the infinite-horizon global cost by conducting  the spatial averaging (of all agents) and the temporal averaging (of infinite horizon) at the same time. 
 
Theoretically, we provide non-asymptotic performance guarantees of ZODPO.
For technical purposes, we consider   static linear  policies, i.e. $u_i(t)=K_i x_{\mathcal I_i}(t)$ for a matrix $K_i$ for  agent $i$, though ZODPO can incorporate more general policies. 

Specifically, we show that, to approach some stationary point with error tolerance $\epsilon$, the required number of samples is $O(n_K^3\max(n,N)\epsilon^{-4})$, where $n_K$ is the dimension of the policy parameter, $N$ is the number of agents and $n$ is the dimension of the state. The polynomial dependence on the problem dimensions indicates the scalability of ZODPO. To the best of our knowledge, this is the first sample complexity result for distributed learning algorithms for the decentralized LQ control considered in this paper. In addition, we  prove that all the policies generated and implemented by ZODPO are stabilizing with high probability, guaranteeing the safety during the learning process.

To establish the sample complexity, compared with the centralized LQR learning setting, we need to bound additional error terms caused by the approximate perturbation sampling and the cost estimation via spatial-temporal averaging. We also point out that existing literature on zero-order-based LQR learning assumes bounded process noises or no noises to guarantee stability \cite{fazel2018global,malik2018derivative}, while this paper allows unbounded Gaussian process noises. To guarantee stability under the unbounded noises, we introduce a truncation step to the gradient estimation to ensure that the estimated gradients are bounded. 
We also explicitly bound the effects of the truncation step when proving the sample complexity.

Numerically, we test  ZODPO  on multi-zone HVAC systems to demonstrate the optimality and safety of the controllers generated by ZODPO.}

\subsection{Related work}\label{subsec:related_work}
There have been numerous studies on related topics including learning-based control, decentralized control, multi-agent reinforcement learning, etc., which are briefly reviewed below.

\begin{enumerate}[label=\alph*),leftmargin=0pt,itemindent=24pt,itemsep=3pt,listparindent=12pt]
\item {\it Learning-based LQ control:} Controller design without (accurate) model information has been studied in the fields of adaptive control \cite{aastrom2013adaptive} and extremum-seeking control \cite{ariyur2003real} for a long time, but most papers focus on stability and asymptotic performance. Recently, much progress has been made on algorithm design and nonasymptotic analysis for learning-based centralized (single-agent) LQ control with full observability, e.g., model-free schemes \cite{fazel2018global,malik2018derivative,yang2019provably}, identification-based controller design \cite{dean2017sample,mania2019certainty}, Thompson sampling \cite{ouyang2017learning}, etc.; and with partial observability \cite{oymak2019non,mania2019certainty}. 
As for learning-based decentralized (multi-agent) LQ control, most studies either adopt a centralized learning scheme \cite{bu2019lqr} or still focus on asymptotic analysis \cite{abouheaf2014multi,zhang2016data,zhang2019online}. Though,   \cite{gagrani2018thompson} proposes a distributed learning algorithm with a
nonasymptotic guarantee, the  algorithm requires  agents to store and
update the model of the whole system, which is prohibitive for large-scale systems.

Our algorithm design and analysis are related to policy gradient for centralized LQ control \cite{fazel2018global,bu2019lqr,malik2018derivative}. Though policy gradient can reach the global optimum in the centralized setting because of the gradient dominance property \cite{fazel2018global}, it does not necessarily hold for decentralized LQ control \cite{feng2019exponential}, and thus we only focus on reaching stationary points as most other papers did in nonconvex optimization \cite{ghadimi2013stochastic}.

\item {\it Decentralized control:}
Even with model information, decentralized control is very challenging. For example, the optimal controller for general decentralized LQ problems may be nonlinear \cite{witsenhausen1968counterexample}, and the computation of such optimal controllers mostly remains unsolved. Even for the special  cases with linear optimal controllers, e.g., the quadratic invariance cases, one usually needs to optimize over an infinite dimensional space \cite{rotkowitz2005characterization}. For tractability, many papers, including this one, consider finite dimensional linear policy spaces and study suboptimal controller design  \cite{maartensson2009gradient,al2011suboptimal}.

\item {\it{Multi-agent reinforcement learning:}}
There are various settings for multi-agent reinforcement learning (MARL),
and our problem is similar to the cooperative setting with partial observability, also known as Dec-POMDP \cite{bernstein2002complexity}. Several MARL algorithms have been developed for Dec-POMDP, including centralized learning decentralized execution approaches, e.g. \cite{omidshafiei2017deep}, and decentralized learning decentralized execution approaches, e.g. \cite{peshkin2000learning,foerster2018counterfactual}. Our proposed algorithm   can be viewed as a decentralized learning decentralized execution approach. 

In addition, most 
cooperative MARL papers for Dec-POMDP assume global cost (reward) signals for agents \cite{peshkin2000learning,omidshafiei2017deep,foerster2018counterfactual}.
However, this paper considers that agents only receive local costs and aim to minimize the averaged costs of all agents. In this sense, our setting is similar to that in \cite{zhang2018fully}, but \cite{zhang2018fully} assumes  global state and global action signals.




\item {\it Policy gradient approaches:}
Policy gradient and its variants are popular  in both RL 
and MARL. 
Various gradient estimation schemes have been proposed, e.g., REINFORCE \cite{williams1992simple}, policy gradient theorem \cite{sutton2000policy}, deterministic policy gradient theorem \cite{silver2014deterministic}, zero-order gradient estimation \cite{fazel2018global}, etc. This paper adopts the zero-order gradient estimation, which has been employed for learning centralized LQ control \cite{fazel2018global,malik2018derivative}.

\item {\it Zero-order optimization:}
It aims to solve optimization without gradients by, e.g., estimating gradients based on function values \cite{flaxman2005online,shamir2013complexity}. This paper adopts the gradient estimator in \cite{flaxman2005online}. {However, due to the distributed setting and communication constraints, we cannot sample policy perturbations exactly as in \cite{flaxman2005online},  since the global objective value is not revealed directly but has to be learned/estimated. Besides, we have to ensure the controllers' stability during the learning procedure, which is an additional requirement not considered in the optimization literature \cite{flaxman2005online}.}


\end{enumerate}

\nit{Notations: }
Let $\|\cdot \|$ denote the $\ell_2$ norm. Let $\|M\|_F$ and $\text{tr}(M)$ denote the Frobenious norm and the trace of a matrix $M$. $M_1\preceq M_2$ means $M_2-M_1$ is positive semidefinite. For a matrix $M$, $\vec(M)$ denotes the vectorization and notice that $\|\vec(M)\|=\|M\|_F$. We will  frequently use the notation:
$$
\vec\big((M_i)_{i=1}^N\big)
=\begin{bmatrix}
\vec(M_1) \\ \vdots \\ \vec(M_N)
\end{bmatrix},
$$
where $M_1,\ldots,M_N$ are arbitrary matrices. We use $\one$ to denote the vector with all one entries, and $I_p \in\R^{p \times p}$ to denote the identity matrix. The unit sphere $\{x\in\mathbb{R}^p: \|x\|=1\}$ is denoted by $\mathbb{S}_p$, and $\mathrm{Uni}(\mathbb{S}_p)$ denotes the uniform distribution on $\mathbb{S}_p$. For any $x\in\mathbb{R}^p$ and any subset $S\subset\mathbb{R}^p$, we denote $x+S\coloneqq \{x+y:y\in S\}$.
The indicator function of a random event $S$ will be denoted by $\mathsf{1}_{S}$ such that $\mathsf{1}_S=1$ when the event $S$ occurs and $\mathsf{1}_S=0$ when $S$ does not occur. 


\section{Problem Formulation}\label{sec:formulation}


Suppose there are $N$ agents jointly controlling a discrete-time linear system of the form
\begin{equation}\label{eq:LTI_system}
    x(t+1) = Ax(t)+Bu(t)+w(t), \quad
    t=0,1, 2, \ldots,
\end{equation}
where $x(t) \in \R^n$ denotes the state vector, $u(t)\in \R^m$ denotes the joint control input, and $w(t)\in \R^n$ denotes the random disturbance at time $t$. We assume $w(0),w(1),\ldots$ are i.i.d. from the Gaussian distribution $\mathcal{N}(0,\Sigma_w)$  for some positive definite matrix $\Sigma_w$. Each agent $i$ is associated with a local control input $u_i(t) \in \R^{m_i}$, which constitutes the global control input $
u(t)=\left[u_1(t)^\top, \ldots,u_N(t)^\top\right]^\top \in \R^n$.

We consider the case where each agent $i$ only observes a partial state, denoted by  $x_{\mathcal I_i}(t)\in \R^{n_i}$, at each time $t$, where $\mathcal I_i$ is a fixed subset of $\{1, \ldots, n\}$ and $x_{\mathcal I_i}(t)$ denotes the subvector of $x(t)$ with indices in $\mathcal I_i$.\footnote{$\mathcal I_i$ and $\mathcal I_{i'}$ may overlap. Our results can be extended to more general observations, e.g., $y_i(t)=C_ix(t)$.}
The admissible local control policies  are limited to the ones that only use the historical local observations. As a starting point, this paper only considers static linear policies that use the current observation, i.e., $u_i(t)=K_i x_{\mathcal{I}_i}(t)$.\footnote{The framework and algorithm can be extended to more general policy classes, but analysis is left as future work.}
For notational simplicity, we define
\begin{equation}
\mK \coloneqq \vec\big((K_i)_{i=1}^N\big)
\in\mathbb{R}^{n_K},
\qquad n_K\coloneqq\sum\nolimits_{i=1}^N n_i m_i.
\end{equation}
It is straightforward to see that the global control policy is also a static linear policy on the current state. We use $\mathcal M(\mK)$ to denote the global control gain, i.e., $u(t)=\mathcal M(\mK) x(t)$. Note that $\mathcal M(\mK)$ is often sparse in network control applications. Figure~\ref{fig:setting_diagram} gives an illustrative example of the our control setup.

\begin{figure}[ht]
\begin{subfigure}{\linewidth}
\centering
\includegraphics[width=.54\textwidth]{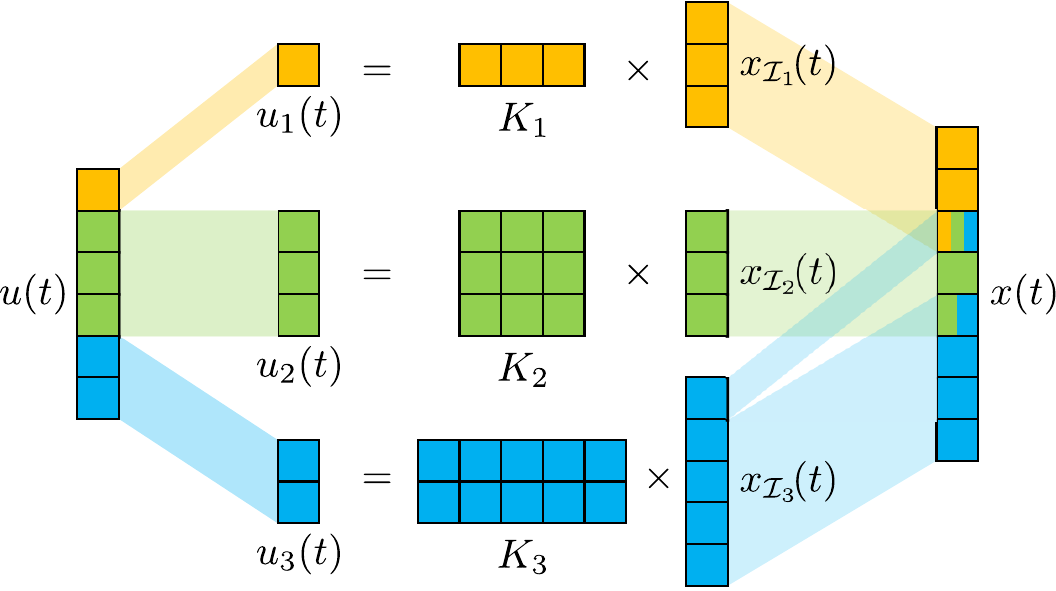}
\end{subfigure}
\begin{subfigure}{\linewidth}
\vspace{8pt}
\centering
\includegraphics[width=.38475\textwidth]{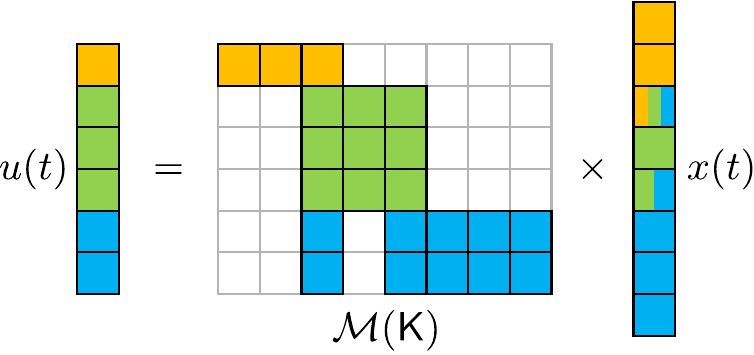}
\end{subfigure}
\caption{An illustrative diagram for $N=3$ agents, where $x(t)\in\mathbb{R}^8$, $u(t)\in\mathbb{R}^6$, and $\mathcal I_1=\{1,2,3\}, \mathcal I_2=\{3,4,5\}, \mathcal I_3=\{3,5,6,7,8\}$. The top figure illustrates the local control inputs,  local controllers, and local observations; and the bottom figure provides a global viewpoint  of the resulting controller $\mathcal M(\mK)$.}
\label{fig:setting_diagram}
\end{figure}

At each time step $t$, agent $i$ receives a quadratic local stage cost $c_i(t)$ given by
$$
c_i(t)=x(t)^\top Q_i x(t)+u(t)^\top R_i u(t),
$$
which is allowed to depend on the global state $x(t)$ and control $u(t)$.
The goal is to find a control policy that minimizes the infinite-horizon average cost among all agents, that is,
\begin{equation}\label{equ: min J(K)}
\begin{aligned}
\min_{\mK} \quad &
J(\mK)\coloneqq
\lim_{T\to \infty}
 \frac{1}{T}
\sum_{t=1}^T\E\!\left[ \frac{1}{N}\sum_{i=1}^N c_i(t) \right] \\
\text{s.t.} \quad
& x(t+1)=Ax(t)+Bu(t)+w(t), \\
&\ \ \ \ 
u_i(t)=K_ix_{\mathcal I_i}(t),
\qquad i=1,\ldots,N.
\end{aligned}
\end{equation}
When the model parameters are known, 
the problem \eqref{equ: min J(K)} can be viewed as a decentralized LQ control problem, which is known to be a challenging problem in general. Various heuristic or approximate methods have been proposed (see Section \ref{subsec:related_work}), but most of them require accurate model information that may be hard to obtain in practice. Motivated by the recent progress in learning based control and also the fact that the models are not well-studied or known for many systems, this paper studies learning-based decentralized control for \eqref{equ: min J(K)}, where each agent $i$ learns the local controller $K_i$ by utilizing the partial states $x_{\mathcal I_i}(t)$ and local costs $c_i(t)$ observed along the system's trajectories.

In many real-world applications of decentralized control, limited communication among agents is available via a communication network. 
Here, we consider a connected and  undirected communication network $\mathcal G\!=\!(\{1, \dots, N\}, \mathcal E)$, where each node represents an agent and  $\mathcal E$ denotes the set of edges. At each time $t$, agent $i$ and $j$ can directly communicate a small number of scalars to each other if and only if $(i,j)\in\mathcal{E}$. Further, we introduce a doubly-stochastic and nonnegative communication matrix $W\!=\![W_{ij}]\!\in\!\R^{N\times N}$ associated with the communication network $\mathcal G$, with $W_{ij}\!=\!0$ if $(i,j)\!\notin\!\mathcal E$ for $i\!\neq\!j$ and $W_{ii}\!>\!0$ for all $i$. The construction of the matrix $W$ has been extensively discussed in  literature (see, for example, \cite{xiao2004fast}). We denote
\begin{equation}
\rho_W\coloneqq\left\|W-\frac{1}{N}\one \one^\top\right\|.
\end{equation}
This quantity captures the convergence rate of the consensus via $W$ and is known to be within $[0,1)$ \cite{xiao2004fast,qu2017harnessing}.

Finally, we introduce the technical assumptions that will be imposed throughout the paper.
\begin{assumption}\label{ass:LQR_basic}
The dynamical system $(A, B)$ is controllable. The cost matrices $Q_i, R_i$ are positive semidefinite for each $i$, and the global cost matrices $\frac{1}{N}\sum_{i=1}^N Q_i$ and $\frac{1}{N}\sum_{i=1}^N R_i$ are positive definite.
\end{assumption}
\begin{assumption}\label{ass: exist K stabilzing}
There exists a control policy $\mK\in\mathbb{R}^{n_K}$ such that the resulting global dynamics
$x(t+1)=(A+B\mathcal M(\mK))x(t)$ is asymptotically stable.
\end{assumption}
Both assumptions are common in LQ control literature. Without Assumption \ref{ass: exist K stabilzing}, the problem \eqref{equ: min J(K)} does not admit a reasonable solution even if all system parameters are known, let alone learning-based control.\footnote{{If Assumption \ref{ass: exist K stabilzing} does not hold but the system is stabilizable, then one has to consider more general controller structures, e.g. linear dynamic controllers, nonlinear controllers, which is beyond the scope of this paper.}}  For ease of exposition, we denote $\mathcal K_{\mathrm{st}}$ as the set of stabilizing controller, i.e.,
$$
\mathcal{K}_{\mathrm{st}}
\coloneqq
\{\mathsf{K}\in\mathbb{R}^{n_K}:
A+B\mathcal M(\mK)\text{ is asymptotically stable}\}.
$$

\section{Algorithm Design}

\subsection{Review: Zero-Order Policy Gradient for Centralized LQR} 

To find a policy $\mK$ that minimizes $J(\mK)$, one common approach is the policy gradient method, that is,
$$  \mK(s+1)=\mK(s)-\eta  \hat{\mathsf{g}}(s),
\quad s=1,2,\ldots,
\qquad
\mK(1)=\mK_0,$$
where $\hat{\mathsf{g}}(s)$ is an estimator of the gradient $\nabla J(\mK(s))$, $\eta>0$ is a stepsize, and $\mK_0$ is some known stabilizing controller. In \cite{fazel2018global} and \cite{malik2018derivative}, the authors have proposed to employ gradient estimators from zero-order optimization. One example is:
\begin{equation}\label{eq:zero_order_grad_est}
\mathsf{G}^r(\mathsf{K},\mathsf{D} ):=\frac{n_K}{r}J(\mK+r\mathsf{D} )\mathsf{D} 
\end{equation}
for  $\mK\in\mathcal{K}_{\mathrm{st}}$ and $r>0$ such that $\mK+r\mathbb{S}_{n_K}\subseteq \mathcal{K}_{\mathrm{st}}$, where $\mathsf D \in \R^{n_K}$ is randomly sampled from $\text{Uni}(\Sb_{n_K})$.
The parameter $r$ is sometimes called the smoothing radius, and it can be shown that the bias $\|\mathbb{E}_{\mathsf{D} }[\mathsf{G}^r(\mathsf{K},\mathsf{D})]-\nabla J(\mathsf{K})\|$ can be controlled by $r$ under certain smoothness conditions on $J(\mathsf{K})$ \cite{malik2018derivative}.
The policy gradient based on the estimator  \eqref{eq:zero_order_grad_est} is given by
\begin{equation}\label{eq:policy_grad_zero_order}
\begin{aligned}
\mK(s+1)=\ &
\mK(s)-\eta\,
\mathsf{G}^r(\mathsf{K}(s),\mathsf{D}(s)) \\
=\ &
\mK(s)-\eta\cdot
\frac{n_K}{r}J(\mathsf{K}(s)+r\mathsf{D}(s))\mathsf{D}(s),
\end{aligned}
\end{equation}
where $\{\mathsf{D}(s)\}_{s=1}^{T_G}$ are i.i.d. random vectors  from $\mathrm{Uni}(\mathbb{S}_{n_K})$.

\subsection{Our Algorithm: Zero-Order Distributed Policy Optimization}

\begin{figure}[t]
\removelatexerror
\begin{algorithm2e}[H]
	\caption{Zero-Order Distributed Policy Optimization (ZODPO)}
	\label{alg:main}
	\SetAlgoNoLine
	\DontPrintSemicolon
	\LinesNumbered
	\KwIn{smoothing radius $r$, step size $\eta$, $\bar J>0$,  $T_G,T_J$,  initial controller $\mK_0\in\mathcal{K}_{\mathrm{st}}$.}
	\SetKwFunction{FCost}{GlobalCostEst}
	\SetKwFunction{FSample}{SampleUSphere}
	
	Initialize $\mathsf{K}(1)= \mathsf{K}_0$.\;

	\For{$s= 1,2,\ldots,T_G$}{
	    \vspace{4pt}
	    \tcp{Step 1: Sampling from the unit sphere}
	    \vspace{2pt}
		Each agent $i$ generates $D_i(s)\in\mathbb{R}^{m_i\times n_i}$ by the subroutine \FSample.\;
		
		\vspace{4pt}
		\tcp{Step 2: Local estimation of the global objective}
		\vspace{2pt}
		Run \FCost{$(K_i(s)+rD_i(s))_{i=1}^N,T_J$}, and let agent $i$'s returned value be denoted by $\tilde{J}_i(s)$.\;

		\vspace{6pt}
			\tcp{Step 3: Local estimation of  partial gradients}
			\vspace{2pt}
				Each agent $i$ estimates the partial gradient $\mfrac{\partial J}{\partial K_i}(\mK(s))$ by 
				$$
				\begin{aligned}
				\hat J_i(s) =\ & \min\!\left\{\tilde{J}_i(s),\bar J\right\}, \\
				\hat{G}^r_i(s)=\ &
				\frac{n_K}{r}	\hat J_i(s)D_i(s)
				.
				\end{aligned}
				$$
				
		\tcp{Step 4: Distributed policy gradient on local controllers}
		\vspace{2pt}
		Each agent $i$ updates $K_i(s+1)$ by
		\begin{equation*}
		K_i(s+1)=K_i(s)-\eta \hat{G}^r_i(s).
		\end{equation*}
	}
\end{algorithm2e}
\vspace{-12pt}
\end{figure}

Now, let us consider the decentralized LQ control formulated in Section \ref{sec:formulation}. Notice that Iteration \eqref{eq:policy_grad_zero_order} can be equivalently written in an almost decoupled way for each agent $i$:
\begin{equation}\label{eq:policy_grad_zero_order_2}
K_i(s+1)=K_i(s)
-\eta\cdot\frac{n_K}{r} J(\mathsf{K}(s)+r\mathsf{D}(s))D_i(s),
\end{equation}
where $\mathsf{D}(s)\sim \operatorname{Uni}(\mathbb{S}_{n_K})$,  $K_i(s), D_i(s)\in \R^{n_i \times m_i } $ and
\begin{equation}\label{equ: vec K, vec D}
\mathsf{K}(s)=\operatorname{vec}((K_i(s))_{i=1}^N), \quad \mathsf{D}(s)=\operatorname{vec}((D_i(s))_{i=1}^N).
\end{equation}
The formulation \eqref{eq:policy_grad_zero_order_2} suggests that, if each agent $i$ can sample $D_i(s)$ properly and obtain the value of the global objective $J(\mathsf{K}(s)+r\mathsf{D}(s))$, then the policy gradient \eqref{eq:policy_grad_zero_order} can be implemented in a decentralized fashion by letting each agent $i$   update its own policy $K_i$ in parallel according to \eqref{eq:policy_grad_zero_order_2}.
This key observation leads us to the ZODPO algorithm (Algorithm \ref{alg:main}).

\begin{figure}[t]
\removelatexerror

\setlength{\interspacetitleruled}{0pt}%
\setlength{\algotitleheightrule}{0pt}%
\begin{algorithm2e}[H]
\SetAlgoNoLine
\DontPrintSemicolon
\LinesNumbered
\SetKwFunction{FSample}{SampleUSphere}
\SetKwProg{Fn}{Subroutine}{:}{}
\Fn{\FSample}{
	Each agent $i$ samples $V_i\in\mathbb{R}^{n_i\times m_i}$ with i.i.d. entries from $\mathcal{N}(0,1)$, and lets $q_i(0)=\|V_i\|_F^2$.\;
	\For{$t= 1,2,\ldots,T_S$}{
		
		Agent $i$ sends $q_i(t-1)$ to its neighbors and updates\vspace{-5pt}
		\begin{equation}\label{eq:sample_usphere_consensus}
		q_i(t)
		=\sum_{j=1}^N W_{ij} q_j(t-1).
		\end{equation}
		\vspace{-5pt}
	}
		\Return{${D}_i:=V_i/\!\sqrt{N q_i(T_S)}$ to agent $i$ for all $i$.}
}
	\vspace{2pt}
\end{algorithm2e}

\setlength{\interspacetitleruled}{0pt}%
\setlength{\algotitleheightrule}{0pt}%
\begin{algorithm2e}[H]
\SetAlgoNoLine
\DontPrintSemicolon
\LinesNumbered
\SetKwFunction{FCost}{GlobalCostEst}
\SetKwProg{Fn}{Subroutine}{:}{}
\Fn{\FCost{$(K_i)_{i=1}^N,T_J$}}{
Reset the system's state  to $x(0)=0$.\;
Each agent $i$ implements $K_i$, and set $\mu_i(0)\leftarrow 0$.\;
\For{$t= 1,2,\ldots,T_J$}{
Each agent $i$ sends $\mu_i(t\!-\!1)$ to its neighbors, observes $c_i(t)$ and updates $\mu_i(t)$ by
\begin{equation}\label{eq:consensus_cost}
\mu_i(t)
=\frac{t-1}{t}\sum_{j=1}^N W_{ij} \mu_j(t-1)
+\frac{1}{t} c_i(t).
\end{equation}
\vspace{-5pt}
}
\Return $\mu_i(T_J)$ to agent $i$ for each $i=1,\ldots,N$.
}
\end{algorithm2e}
\vspace{-12pt}
\end{figure}

Roughly speaking, ZODPO conducts distributed policy gradient iterations with four main steps:

\begin{itemize}[leftmargin=10pt,itemsep=2pt,parsep=2pt,listparindent=12pt]

\item In \textit{Step 1}, each agent $i$ runs the subroutine {\tt SampleUSphere} to generate a random matrix $D_i(s)\in\mathbb{R}^{m_i\times n_i}$ so that the concatenated $\mathsf{D}(s)$ approximately follows the uniform distribution on $\mathbb{S}_{n_K}$. In the subroutine {\tt SampleUSphere}, each agent $i$ samples a Gaussian random matrix $V_i$ independently, and then employs a simple consensus procedure \eqref{eq:sample_usphere_consensus} to compute the averaged squared norm $\frac{1}{N}\sum_{i}\|V_i\|_F^2$. Our analysis shows that the outputs of the subroutine approximately follow the desired distribution for sufficiently large $T_S$ (see Lemma~\ref{lemma:sampling_error} in Section~\ref{subsec:sampling_error}).


\item In \textit{Step 2}, each agent $i$ estimates the  global objective  $J(\mK(s)+r\mD(s))$ by implementing the local policy $K_i(s)+rD_i(s)$ and executing the subroutine {\tt GlobalCostEst}. The subroutine {\tt GlobalCostEst} allows the agents to form local estimates of the global objective value from observed local stage costs and communication with neighbors. Specifically, given the input controller $\mK$ of {\tt GlobalCostEst}, the quantity $\mu_i(t)$ records agent $i$'s estimation of $J(\mK)$ at time step $t$, and is updated based on its neighbors' estimates $\mu_j(t-1)$ and its local stage cost  $c_i(t)$.  The updating rule \eqref{eq:consensus_cost} can be viewed as a combination of a consensus procedure via the communication matrix $W$ and an online computation of the average $\frac{1}{t}\sum_{\tau=1}^{t} c_i(\tau)$. Our  theoretical analysis  justifies that $\mu_i(T_J)\approx J(\mK)$ for sufficiently large $T_J$ (see Lemma~\ref{lem: bias and var of mui}).

Note that the consensus \eqref{eq:sample_usphere_consensus} in the subroutine {\tt SampleUSphere} can be carried out simultaneously with the consensus \eqref{eq:consensus_cost} in the subroutine {\tt GlobalCostEst} as the linear system evolves, in which case $T_S=T_J$. We present the two subroutines separately for clarity.

\item In {\it Step 3}, each agent $i$ forms its partial gradient estimation $\hat{G}_i^r(s)$ associated with its local controller. The partial gradient estimation $\hat{G}_i^r(s)$ is based on \eqref{eq:policy_grad_zero_order_2}, but uses local estimation of the global objective instead of its exact value. We also introduce a truncation step $\hat J_i(s) =\min\!\left\{\mu_i(T_J),\bar J\right\}$ for some sufficiently large $\bar J$, which guarantees the boundedness of the gradient estimator in Step 2 to help ensure the stability of our iterating policy $\mK(s+1)$ and simplify the  analysis.

\item In \textit{Step 4}, each agent $i$ updates its local policy $K_i$ by  \eqref{eq:policy_grad_zero_order_2}.
\end{itemize}


We point out that, per communication round, each agent $i$ only shares a scalar $\mu_i(t)$ for global cost estimation in {\tt GlobalCostEst} and a scalar $q_i(t)$ for jointly sampling in {\tt SampleUSphere},   demonstrating the applicability  in the limited-communication scenarios. Besides, each agent $i$ only stores and updates the local policy $K_i$, indicating that only small storage is used even in large-scale systems.

\begin{remark}
 ZODPO conducts large enough ($T_J$ and $T_S$) subroutine iterations for each policy gradient update (see Theorem \ref{theorem:main} in Section \ref{sec:main}). In practice, one may prefer fewer subroutine iterations, e.g. actor-critic algorithms. However, the design and  analysis of  actor-critic algorithms for our problem are non-trivial since we have to ensure stability/safety during the learning. Currently, ZODPO requires large enough subroutine iterations for good estimated gradients, so that the policy gradient updates do not drive the policy outside the stabilizing region. To overcome this challenge, we consider employing a safe policy and switching to the safe policy whenever the states are too large and resuming the learning when the states are small. In this way, we can use fewer subroutine iterations and ensure safety/stability even with poorer estimated gradients. The theoretical analysis for this method is left as future work.

\end{remark}

\section{{Theoretical Analysis}}\label{sec:main}
In this section, we first discuss some properties of $J(\mK)$, and then provide the  nonasymptotic performance guarantees of ZODPO, followed by some discussions.

As indicated by \cite{feng2019exponential,bu2019lqr}, the objective function $J(\mK)$ of decentralized LQ control can be nonconvex. Nevertheless, $J(\mathsf{K})$ satisfies some smoothness properties.

\begin{lemma}[Properties of $J(\mK)$]\label{lemma:J_smoothness}
The function $J(\mK)$ has the following properties:
\begin{enumerate}[leftmargin=13pt]
\item $J(\mathsf{K})$ is continuously differentiable over $\mathsf{K}\in\mathcal{K}_{\mathrm{st}}$. In addition, any nonempty sublevel set
$
\mathcal Q_\alpha:=\{\mathsf{K}\in\mathcal{K}_{\mathrm{st}}:
J(\mathsf{K})\leq\alpha\}
$
is compact.
\item Given a nonempty sublevel set $\mathcal{Q}_{\alpha_1}$ and an arbitrary $\alpha_2>\alpha_1$, there exist constants $\xi>0$ and $\phi>0$ such that, for any $\mathsf{K}\in \mathcal Q_{\alpha_1}$ and  $\mathsf{K}'$ with $\|\mathsf{K}'-\mathsf{K}\|\leq\xi$, we have
$
\mK'\in \mathcal Q_{\alpha_2}$ and $
\|\nabla J(\mathsf{K}')-\nabla J(\mathsf{K})\|\leq\phi
\|\mathsf{K}'-\mathsf{K}\|.$
\end{enumerate}
\end{lemma}

This lemma is essentially  \cite[Lemma~7.3 \& Corollary~3.7.1]{bu2019lqr} and \cite[Lemmas~1~\&~2]{malik2018derivative}. Without loss of generality, we let
$$
\begin{aligned}
\mathcal{Q}^0=\ &
\{\mathsf{K}\in\mathcal{K}_{\mathrm{st}}:
J(\mathsf{K})\leq 10J(\mK_0)\}, \\
\mathcal{Q}^1=\ &
\{\mathsf{K}\in\mathcal{K}_{\mathrm{st}}:
J(\mathsf{K})\leq 20J(\mK_0)\}.
\end{aligned}
$$
Lemma~\ref{lemma:J_smoothness} then guarantees that there exist $\xi_0>0$ and $\phi_0>0$ such that for any $\mK\in\mathcal{Q}^0$ and any $\mK'$ with $\|\mK'-\mK\|\leq\xi_0$, we have $\mK'\in\mathcal{Q}^1$ and $\|\nabla J(\mK')-\nabla J(\mK)\|\leq\phi_0$. The constants $\xi_0$ and $\phi_0$ depend on $A$, $B$, $\Sigma_w$,  $J(\mathsf{K}_0)$ and $Q_i$, $R_i$ for all $i$.

With the definitions of $\xi_0, \phi_0$ above, we are ready for the performance guarantee of our ZODPO.

\begin{theorem}[Main result]\label{theorem:main}
Let $\mK_0 \in \mathcal K_{\mathrm{st}}$ be an arbitrary initial controller. Let $\epsilon>0$ be sufficiently small, and suppose
$$
\begin{aligned}
r\leq\ &
\frac{\sqrt{\epsilon}}{40\phi_0}, 
\qquad
\bar J\geq 50J(\mK_0),
\qquad 
\eta \leq
\min
\left\{\frac{14\xi_0 r}{15\bar J n_K},
\frac{3\epsilon r^2}{320\phi_0(40J(\mathsf{K}_0))^2\cdot  n_K^2}\right\}, \\
T_J
\geq\ &
10^3
J(\mathsf{K}_0)\frac{n_K}{ r\sqrt{\epsilon}}\max\left\{
n\beta_0^2,
\frac{N}{1\!-\!\rho_W}\right\},
\qquad
T_S\geq \frac{\log\frac{8N^2}{\phi_0\eta}}{-2\log\rho_W}, \\
T_G =\ &
c\cdot\frac{40J(\mathsf{K}_0)}{\eta\epsilon}, \quad \frac{1}{16}\leq c\leq 16,
\end{aligned}
$$
where  $\beta_0$ is a constant determined by  $A, B, \Sigma_w, \mK_0$ and $Q_i, R_i$ for all $i$. Then, the following two statements hold.
\begin{enumerate}[leftmargin=13pt]
\item The controllers $\{\mK(s)\}_{s=1}^{T_G}$ generated by Algorithm \ref{alg:main} are all stabilizing  with probability at least $0.9- 0.05c$.
\item The controllers $\{\mK(s)\}_{s=1}^{T_G}$  enjoy the  bound below with probability at least $0.875- 0.05(c+c^{-1})$:
\begin{equation}\label{equ: main_thm_convergence}
\frac{1}{T_G} \sum_{s=1}^{T_G}\left\|\nabla J(\mK(s))\right\|^2 \leq \epsilon.
\end{equation}
Further, if we select $\hat{\mathsf{K}}$  uniformly randomly  from $\{\mathsf{K}(s)\}_{s=1}^{T_G}$, then  with probability at least $0.875- 0.05(c+c^{-1})$,
\begin{equation}\label{equ: main_thm_convergence_2}
\Big\|\nabla J(\hat{\mathsf{K}})\Big\|^2 \leq \epsilon.
\end{equation}
\end{enumerate}
\end{theorem}

The proof is deferred to Section~\ref{sec:proof}. In the following, we provide some discussions regarding Theorem~\ref{theorem:main}.





\begin{itemize}[leftmargin=10pt,topsep=2pt,itemsep=2pt,listparindent=12pt]


\item \textbf{Probabilistic guarantees.} Theorem \ref{theorem:main} establishes the stability and optimality of the controllers generated by ZODPO in a ``with high probability'' sense. 
The variable $c$ in the probability  bounds represents the value of $T_G$ since $T_G =
c\cdot 40J(\mathsf{K}_0)/(\eta\epsilon)$.

Statement 1 suggests that as $T_G$ increases, the probability that all the generated controllers are stabilizing will decrease. Intuitively, this is because the ZODPO can be viewed as a stochastic gradient descent, and as $T_G$ increases, the biases and variances of the gradient estimation  accumulate, resulting in a larger probability of generating destabilizing controllers.

Statement 2 indicates that as $T_G$ increases, the probability of enjoying the optimality guarantees \eqref{equ: main_thm_convergence} and \eqref{equ: main_thm_convergence_2} will first increase and then decrease. This is a result of the trade-off between a higher chance of generating destabilizing controllers and improving the policies by more policy gradient iterations as $T_G$ increases. In other words, if $T_G$ is too small, more iterations will improve the performance of the generated controllers; while for large $T_G$, the probability of generating destabilizing controllers becomes dominant. 


Finally, we mention that the probability bounds are not restrictive and can be improved by, e.g., increasing the numerical factors of $T_J$, using smaller stepsizes, or by  the repeated learning tricks described in \cite{ghadimi2013stochastic,malik2018derivative}, etc.

\item \textbf{Output controller.}
Due to the nonconvexity of $J(\mK)$, we evaluate the algorithm performance by the averaged squared norm of the gradients  of $\{\mathsf{K}(s)\}_{s=1}^{T_G}$ in \eqref{equ: main_thm_convergence}. Besides, we also consider an output controller that is uniformly randomly selected from $\{\mathsf{K}(s)\}_{s=1}^{T_G}$, and provide its performance guarantee \eqref{equ: main_thm_convergence_2}. Such approaches are common in nonconvex optimization \cite{ghadimi2013stochastic,reddi2016stochastic}. Our numerical experiments suggest that selecting $\mK(T_G)$ also yields satisfactory performance in most cases (see Section~\ref{sec:simulation}).

\item \textbf{Sample complexity.} 
The  number of samples to guarantee  \eqref{equ: main_thm_convergence}  with high probability is given by
\begin{equation}\label{eq:sample_complexity}
 T_G T_J= \Theta\left(
\frac{n_K^3}{\epsilon^4}
\max\left\{n\beta_0^2,\frac{N}{1-\rho_W}\right\}\right),
\end{equation}
where we apply 
 the equality conditions in Theorem \ref{theorem:main} and neglect the numerical constants since they are conservative and not restrictive. Some discussions are provided below.
\begin{itemize}[leftmargin=12pt]
\item The sample complexity \eqref{eq:sample_complexity}  has an explicit polynomial dependence on the error tolerance's inverse $\epsilon^{-1}$,  the number of controller parameters $n_K$ and the number of agents $N$, demonstrating the scalability of ZODPO.

\item The sample complexity depends on the maximum of the two terms: (i) term $N/(1\!-\!\rho_W)$ stems from the consensus procedure among $N$ agents, which increases with $\rho_W$ as a larger $\rho_W$ indicates a smaller consensus rate; (ii) term $n\beta_0^2$ stems from  approximating  the infinite-horizon averaged cost, which exists even for a single agent. 

\item Notice that \eqref{eq:sample_complexity} is proportional to $n_K^3$. Detailed analysis reveals that the variance of the single-point gradient estimation contributes a dependence of $n_K^2$, which also accords with the theoretical lower bound for zero-order optimization in \cite{shamir2013complexity}. The additional $n_K$ comes from the non-zero bias of the global cost estimation.


\item While there is an explicit linear asymptotic dependence on the state vector dimension $n$ in \eqref{eq:sample_complexity}, we point out that the quantities $\beta_0, J(\mK_0), \phi_0, \xi_0$ are also implicitly affected by $n$ as they are determined by $A, B, Q, R, \Sigma_w$ and $\mK_0$. Thus, the actual dependence  on $n$ is complicated and not straightforward to summarize.
\end{itemize}

\item \textbf{Optimization landscape.} Unlike centralized LQ control with full observations, reaching the global optimum is extremely challenging for general decentralized LQ control with partial observations. In some cases,   the stabilizing region $\mathcal K_{\text{st}}$  may even contain multiple connected components \cite{feng2019exponential}. However, ZODPO  only explores the  component  containing the initial controller  $\mK_0$, so $\mK_0$  affects which stationary points  ZODPO converges to.  How to initialize $\mK_0$ and explore other components effectively based on prior or domain knowledge remain challenging problems and are left as future work.



\end{itemize}

\section{{Proof of Theorem \ref{theorem:main}}}
\label{sec:proof}

This section provides the proof of Theorem \ref{theorem:main}. We  introduce necessary notations,  outline the main ideas of the proof,  remark on the differences between our proof and the proofs in related literature \cite{fazel2018global,malik2018derivative}, and then  provide proof details in subsections.

\vspace{5pt}

\nit{Notations.} In the following, we introduce some useful notations for any stabilizing controller $\mK \in \mathcal{K}_{\mathrm{st}}$.
First, we let $\tilde J_i(\mK)$ denote agent $i$'s estimation of the global objective $J(\mK)$ through the subroutine {\tt GlobalCostEst}, i.e.,
\begin{align*}
\big( \tilde J_1(\mK) ,\dots, \tilde J_N(\mK)\big)& \coloneqq \texttt{GlobalCostEst}\big((K_i)_{i=1}^N, T_J\big),
\end{align*}
and we let $\hat J_i(\mK) \coloneqq  \min \big\{ \tilde J_i(\mK), \bar J\big\}$ denote the truncation of $\tilde J_i(\mK)$. Notice that $\tilde J_i(\mK(s))$ and $\hat J_i(\mK(s))$ correspond to $\tilde J_i(s)$ and $\hat J_i(s)$ in Algorithm \ref{alg:main} respectively.
Then, for any $r>0$ and $\mD \in \mathbb{R}^{n_K}$ such that $\mathsf{K}+r\mD\in \mathcal{K}_{\mathrm{st}}$, we define
\begin{align*}
\hat{G}^r_i(\mathsf{K},\mathsf{D})
\coloneqq\ &
\frac{n_K}{r}\hat{J}_i(\mK \!+\! r\mD)D_i, \quad \forall \, 1 \leq i \leq N, \\
\hat{\mathsf{G}}^r(\mathsf{K},\mathsf{D}) \coloneqq \ & \vec\big((\hat{G}^r_i(\mathsf{K},\mathsf{D}))_{i=1}^N\big),
\end{align*}
where $D_i$ are $m_i\times n_i$ matrices such that
$
\mathsf{D}= \vec\big((D_i)_{i=1}^N\big)$.
Notice that $\hat {G}^r_i(\mK, \mD)$ denotes agent $i$'s estimate of the partial gradient $\mfrac{\partial J}{\partial K_i}(\mK)$ given the controller $\mK$ and perturbation $\mD$. In particular, $\hat {G}^r_i(\mK(s), \mD(s))$ corresponds to $\hat G^r_i(s)$ in Step~3 of Algorithm \ref{alg:main}. The vector $\hat{\mathsf{{G}}}^r(\mK, \mD)$ that concatenates all the (vectorized) partial gradient estimates of the agents then gives an estimate of the complete gradient vector $\nabla J(\mK)$.

\vspace{5pt}

\nit{Proof Outline.} Our proof mainly consists of six parts.
\begin{enumerate}[label=(\alph*),leftmargin=18pt,labelsep=4pt,labelwidth=11pt,itemsep=0pt,parsep=1pt,topsep=2pt]
    \item Bound the sampling error in Step 1 of Algorithm~\ref{alg:main}.
    \item Bound the estimation error of the global objective generated by Step~2 of Algorithm~\ref{alg:main}.
    \item Bound the estimation error of partial gradients generated by Step~3 of Algorithm~\ref{alg:main}.
    \item Characterize the improvement by one-step distributed policy update in Step~4 of Algorithm~\ref{alg:main}.
    \item Prove Statement~1 in Theorem~\ref{theorem:main}, i.e. all the generated controllers are stabilizing with high probability.
    \item Prove Statement~2 in Theorem~\ref{theorem:main}, i.e. the bound \eqref{equ: main_thm_convergence}.
\end{enumerate}
Each part is discussed in detail in the subsequent subsections.

\begin{remark}
{Our proof is inspired by the zero-order-based centralized LQR learning literature \cite{fazel2018global,malik2018derivative}. Below, we remark on the major differences between our proofs and \cite{fazel2018global,malik2018derivative}.

Firstly, unlike \cite{fazel2018global,malik2018derivative}, we cannot sample from the uniform sphere distribution exactly due to the distributed setting. Therefore, we have to bound the errors of the approximate sampling method (Lemma \ref{lemma:sampling_error}).

Secondly, \cite{fazel2018global,malik2018derivative} assume no process noises or bounded noises, while this paper assumes unbounded Gaussian noises. To ensure stability during the learning  under the unbounded noises, we introduce a truncation step in Line 5 of Algorithm \ref{alg:main}. Thus, we need to bound the truncation errors (Lemma \ref{lemma:capped_J_bias}).

Thirdly, when bounding the cost estimation errors, \cite{fazel2018global}, \cite{malik2018derivative} only need to address the temporal averaging, while we address both temporal averaging and spatial averaging (Lemma \ref{lem: bias and var of mui}).

Lastly, the objective function in the centralized LQR is gradient dominant \cite{fazel2018global}, but our objective function lacks such a nice property due to the partial observation and the decentralized control  structure. Therefore, we have to rely on general nonconvex  analysis and more conservative choices of parameters to bound the accummulated errors in  Section~\ref{subsec:proof_final_part}.}

\end{remark}

\subsection{Bounding the Sampling Inaccuracy}
\label{subsec:sampling_error}
In this part, we focus on the subroutine {\tt SampleUSphere} and bound the deviation of it outputs from the desired distribution $\mathrm{Uni}(\mathbb{S}_{n_K})$.

\begin{lemma}\label{lemma:sampling_error}
Consider the subroutine {\tt SampleUSphere},  let $$
D_i^{0} = \frac{V_i}{\sqrt{\sum_{i=1}^N \|V_i\|_F^2}}
\qquad\text{and}\qquad \mD^0=\vec\big((D^0_i)_{i=1}^N\big).
$$
Then, we have    $\mD^0 \sim\mathrm{Uni}(\mathbb{S}_{n_K})$.
Further, for
$
T_S\geq \log {2N}/(-\log\rho_W)
$,
\begin{equation}\label{eq:sampling_error_bound}
\begin{aligned}
\left\|D_i-D_i^{0}\right\|_F
\leq\ &
N\rho_W^{T_S}\cdot\|D_i^0\|_F,
\quad i=1,\ldots,N, \\
\sum_{i=1}^N\|{D}_i\|_F^2
\leq\ &
\left(1+N\rho_W^{T_S}\right)^2.
\end{aligned}
\end{equation}
\end{lemma}


The proof is in Appendix \ref{append: sampling error}.

\subsection{Bounding the Global Cost Estimation Error}
\label{subsec:proof_part1}

In this part, we bound the difference between the global cost estimation $\tilde{J}_i(\mK)$ and  the true cost $J(\mK)$ for any $\mK \in \Q^1$. Besides, we bound the expected difference between $\tilde J_i(\mK)$ and the truncated estimation $\hat J_i(\mK)$. Later in Section~\ref{subsec:proof_stability}, we will show that the outputs generated by Algorithm~\ref{alg:main} are inside $\Q^0 \subseteq \Q^1$ with high probability, thus the bounds here characterize the properties of the output controllers.

\begin{lemma}[Estimation error of {\tt GlobalCostEst}]\label{lem: bias and var of mui}
There exists $\beta_0>0$ determined by $A, B, Q_i, R_i, \mK_0, \Sigma_{w}$, such that for  any $\mK \in \Q^1$
and any $1\leq i\leq N$, 
\begin{align}
\left|\mathbb{E}\big[\tilde{J}_i(\mK)\big]
\!-\!
J(\mathsf{K})
\right|
\leq\,&
\frac{J(\mathsf{K})}{T_J}
\left[
\frac{N}{1-\rho_W}
+\beta_0
\right],\label{eq:bias_mu_i} \\
\mathbb{E}
\!
\big(\tilde{J}_i(\mK) \!-\! J(\mathsf{K})\big)^2
\leq\,&
\frac{6nJ(\mathsf{K})^2}{T_J}
\beta_0^2
+\!
\frac{8J(\mathsf{K})^2}{T_J^2}
\!\left[\frac{N}{1\!-\!\rho_W}\right]^{\!2}\! \label{eq:square_diff_mu_i_J}
\end{align}
where the expectation is taken with respect to the process noises when implementing {\tt GlobalCostEst}.
\end{lemma}
The proof is in Appendix \ref{append: bias and var of cost estimation}.

\begin{lemma}[Effect of truncation]\label{lemma:capped_J_bias}
Given $\mK \in \Q^1$, when the constant $\bar J$ satisfies
$$
\bar{J}\geq 
J(\mathsf{K}) \max\left\{\frac{5}{2},
\frac{5N}{T_J(1-\rho_W)}
\right\},
$$
the truncation $\hat J_i(\mK)=\min(\tilde J_i(\mK), \bar J)$ will satisfy
$$
0\leq \mathbb{E}\big[
\tilde{J}_i(\mK)
-\hat{J}_i(\mK)
\big]
\leq
\frac{90 J(\mathsf{K})}{T_J^2}
\left[
n^2\beta_0^4
+\frac{N^2}{(1-\rho_W)^2}
\right].
$$
where $\beta_0$ is defined in Lemma \ref{lem: bias and var of mui}, and  the expectation is taken with respect to the process noises when implementing {\tt GlobalCostEst}.
\end{lemma}
The proof is in Appendix \ref{append: truncation error}.


\subsection{Bounding the Gradient Estimation Error}
In this part, we bound the bias and the second moment of the gradient estimator $\hat{\mathsf{G}}^r(\mK, \mD)$ for any $\mK \in \Q^0$.
Our $\hat{\mathsf{G}}^r(\mK, \mD)$ is based on the zero-order gradient estimator $\mathsf{G}^r(\mK, \mD)$ defined in \eqref{eq:zero_order_grad_est}, whose bias can be bounded by the following lemma. 
\begin{lemma}[{\hspace{1sp}\cite[Lemma 6]{malik2018derivative}}]\label{lem: grad Jr(k)-J(K) bdd}
Consider any  $\mathsf{K}\in \mathcal{Q}^0$ and $\mathsf{D}\sim\mathrm{Uni}(\mathbb{S}_{n_K})$, then
$
\left\|\mathbb{E}_{\mathsf{D}}[\mG^r(\mK, \mD)]
-\nabla J(\mathsf{K})\right\|
\leq \phi_0 r
$
for  $r\leq\xi_0$.
\end{lemma}

Notice that our gradient estimator $\hat{\mathsf{G}}^r(\mK, \mD)$ relies on the estimated objective value $\hat J_i(\mK+r \mD)$ instead of the accurate value $J(\mK+r\mD)$ as in ${\mathsf{G}}^r(\mK, \mD)$; in addition, the distribution of $\mD$ is only an approximation of $\mathrm{Uni}(\mathbb{S}_{n_K})$. Consequently, there will be additional error in the gradient estimation step. By leveraging Lemma \ref{lem: grad Jr(k)-J(K) bdd} and the cost estimation error bounds in Lemmas~\ref{lem: bias and var of mui} and \ref{lemma:capped_J_bias}
in Section~\ref{subsec:proof_part1}, we obtain bounds on the bias and second moment of our gradient estimator $\hat{\mathsf{G}}^r(\mK, \mD)$.

We introduce an auxiliary quantity $\kappa_0\coloneqq\sup_{\mK\in\mathcal{Q}^1}\|\nabla J(\mK)\|
$.
Lemma~\ref{lemma:J_smoothness} guarantees that $\kappa_0<+\infty$.
\begin{lemma}[Properties of gradient estimation]\label{lem:bias_var_grad_est}
Let $\delta\in(0,1/14]$ be arbitrary. Suppose
\begin{align*}
r & \leq
\min\left\{\frac{14}{15}\xi_0,
\frac{20J(\mK_0)}{\kappa_0}\right\},\quad
\bar J \geq  50 J(\mathsf{K}_0),\\
T_J & \geq
120\max\left\{
n\beta_0^2,
\frac{N}{1-\rho_W}\right\},
\quad
T_S \geq
\frac{\log(N/\delta)}{-\log\rho_W},
\end{align*}
and let $\mathsf{D}$ be generated by {\tt SampleUSphere}. Then for any $\mK \in \mathcal{Q}^0$, we have $\mK+r\mathsf{D}\in\mathcal{Q}^1$. Furthermore,
\begin{align}
\left\|\E\!\left[\hat{\mathsf{G}}^r(\mK,\mathsf{D})\right]
\!-\! \nabla J(\mK)\right\|^2 
\leq\ &
5\phi_0^2r^2
+2\left(\frac{50\delta n_K J(\mK_0)}{r}\right)^2
+
5\left(\frac{50n_K J(\mathsf{K}_0)}{r T_J}
\max\left\{n\beta_0^2
,\frac{N}{1-\rho_W}\right\}\right)^2\label{equ: grad error bdd}\\
\mathbb{E}\!
\left[
\big\|\hat{\mathsf{G}}^r(\mathsf{K},\mathsf{D})\big\|^2
\right]
\leq\ &
\!\left(30J(\mathsf{K}_0)\frac{n_K}{r}\right)^2,\label{equ: grad 2 moment bdd}
\end{align}
where the expectation $\E$ is with respect to $\mathsf{D}$ and the system process noises  in the subroutine {\tt GlobalCostEst}.
\end{lemma}
\begin{proof}
Firstly, the condition on $T_S$ implies $N\rho_W^{T_S}\leq \delta$ since $\rho_W<1$, and by Lemma~\ref{lemma:sampling_error}, we have
\begin{equation}\label{equ: |D|<1+delta}
\|\mathsf{D}\|^2
=\sum_{i=1}^N \|D_i\|_F^2
\leq (1+\delta)^2
\leq\left(\frac{15}{14}\right)^2,
\end{equation}
and consequently $\|r\mathsf{D}\|\leq \xi_0$. By the definition of $\xi_0$ below Lemma \ref{lemma:J_smoothness}, we have that $\mK+r \mD \in \Q^1$ for any $\mK \in \Q^0$.

We then proceed to prove the two inequalities. We let $D_i^0,\,i=1,\ldots,N$ and $\mathsf{D}^0$ denote the random matrices and the random vector as defined in Lemma~\ref{lemma:sampling_error}, so that $\mathsf{D}^0\sim\mathrm{Uni}(\mathbb{S}_{n_K})$ and the bounds \eqref{eq:sampling_error_bound} hold.

\begin{itemize}[leftmargin=0pt, itemindent=12pt]

\item {\it Proof of \eqref{equ: grad error bdd}:}
Let $\mK\in\mathcal{Q}^0$ be arbitrary. Notice that
\begin{align*}
\left\|\E\!\left[\hat{\mathsf{G}}^r(\mK,\mathsf{D})\right]- \nabla J(\mK)
\right\|^2
\leq\,&
\frac{5}{4}
\!\left\|\E\!\left[\hat{\mathsf{G}}^r\!(\mK,\!\mathsf{D})
\!-\!
\mathsf{G}^r\!(\mK,\!\mathsf{D}^0)
\right]\right\|^2
\!\!+\! 5\!\left\|
\mathbb{E}\!\left[\mathsf{G}^r\!(\mK,\!\mathsf{D}^0)\right] \!-\! \nabla J(\mK)\right\|^2\\
\leq\,&
\frac{5}{4}\frac{n_K^2}{r^2}\sum_{i=1}^N \left\|\E\!\left[ \hat{J}_i(\mK \!+\! r\mD)D_i
-J(\mathsf{K} \!+\! r\mathsf{D}^0)D_i^0\right]\right\|_F^2
+ 5\phi_0^2 r^2,
\end{align*}
where we use Lemma \ref{lem: grad Jr(k)-J(K) bdd} and $2 xy \leq (x/r)^2 +(ry)^2$ for any $x,y,r$. By leveraging the same trick, we  bound  the first term:
\begin{align}
& \sum_{i=1}^N\left\|
\E\!\left[ \hat{J}_i(\mK + r\mD)D_i
-J(\mathsf{K} + r\mathsf{D}^0)D_i^0\right]
\right\|_F^2 \nonumber\\
\leq\ &
\frac{5}{3}\sum_{i=1}^N\left\|
\E\!\left[\left(\hat{J_i}(\mK+r\mD)D_i
-J(\mK+r\mD)\right)D_i\right]\right\|^2_F \label{equ: P1}\tag{P1}\\
& +
5\sum_{i=1}^N\left\|\E\!\left[
J(\mK+r\mD)\left(D_i-D_i^0\right)\right]\right\|^2_F \label{equ: P2}\tag{P2}\\
& + 
5\sum_{i=1}^N\left\|\E\!\left[
\left(J(\mK+r\mD)
-J(\mathsf{K}+r\mathsf{D}^0)\right)D_i^0\right]
\right\|_F^2.\label{equ: P3}\tag{P3}
\end{align}
Next, we bounds \eqref{equ: P1}, \eqref{equ: P2}, \eqref{equ: P3}. Remember that $K+r D\in \Q^1$.
For \eqref{equ: P1}, by \eqref{equ: |D|<1+delta}, we have
\begin{align*}
\sum_{i=1}^N\left\|
\E\!\left[\left(\hat{J_i}(\mK+r\mD)D_i
-J(\mK+r\mD)\right)D_i\right]\right\|^2_F
=\,&
\sum_{i=1}^N \left\|\E_\mD
\!\left[
\E\!\left[
\left.\hat J_i(\mK \!+\! r\mD)-J(\mK\!+\!r\mD)
\right|\mD\right] D_i \right]\right\|_F^2\\
\leq\,&
\sum_{i=1}^N
\sup_{\mK' \in \mathcal{Q}^1}
\!\!
\left\{\left|\E_w\!\left[ \hat J_i(\mK')\right] \!-\! J(\mK')\right|^2\right\}
\!\cdot\! \E\!\left[\|D_i\|_F^2\right]\\
\leq\,&
(1+\delta)^2 \max_{1\leq i \leq N}
\sup_{\mK' \in \mathcal{Q}^1}
\!\!\left\{\left|\E_w\!\left[ \hat J_i(\mK')\right]-J(\mK')\right|^2\right\},
\end{align*}
where $\mathbb{E}_w$ denotes expectation with respect to the noise process of the dynamical system in the subroutine {\tt GlobalCostEst}. For any $i$ and $\mK'\in \Q^1$, we have
$ \left|\E_w\!\left[\hat J_i(\mK')\right]-J(\mK')\right| 
\leq
\left|\E_w\!\left[\hat J_i(\mK')-\E \tilde J_i(\mK')\right] \right|
+ \left|\E_w\!\left[\tilde J_i(\mK')\right]-J(\mK')\right|
\leq
\frac{90 J(\mathsf{K}')}{T_J^2}
\left(
n^2
\beta_0^4
\!+\!\frac{N^2}{(1-\rho_W)^2}
\right)
+
\frac{J(\mathsf{K}')}{T_J}\left(
\beta_0
\!+\!
\frac{N}{1-\rho_W}\right),
$
where we used Lemma \ref{lemma:capped_J_bias} and the bound \eqref{eq:bias_mu_i} in Lemma \ref{lem: bias and var of mui} in the second inequality. Notice that the condition on $T_J$ implies
$$
\frac{1}{T_J^2}
\left(n^2\beta_0^4+\frac{N^2}{(1-\rho_W)^2}\right)
\leq
\frac{1}{120T_J}\left(n\beta_0^2
+\frac{N}{1-\rho_W}\right),
$$
and since $J(\mathsf{K}')\leq 20 J(\mathsf{K}_0)$ for $\mathsf{K}'\in\mathcal{Q}^1$, we obtain
$$
\begin{aligned}
\left|\E_w\!\left[\hat J_i(\mK')\right]-J(\mK')\right|
\leq\,&
\frac{3 J(\mathsf{K}')}{4T_J}
\left(
n
\beta_0^2
\!+\!\frac{N}{1-\rho_W}
\right)
+
\frac{J(\mathsf{K}')}{T_J}\left(
\beta_0
\!+\!
\frac{N}{1-\rho_W}\right) \\
\leq\,&
\frac{20J(\mathsf{K}_0)}{T_J}
\left(
\left(\frac{3}{4}n+1\right)\beta_0^2
+\left(\frac{3}{4}+1\right)
\frac{N}{1-\rho_W}
\right) \\
\leq\,&
\frac{35J(\mathsf{K}_0)}{T_J}
\!\left(\!n\beta_0^2
\!+\!\frac{N}{1\!-\!\rho_W}\!\right)
\!\leq\!
\frac{70J(\mathsf{K}_0)}{T_J}
\max\!\left\{\!n\beta_0^2
,\frac{N}{1\!-\!\rho_W}\!\right\}.
\end{aligned}
$$
Therefore, by $\delta\leq 1/14$, we have
$$
\begin{aligned}
\sum_{i=1}^N\left\|
\E\!\left[\left(\hat{J_i}(\mK+r\mD)D_i
-J(\mK+r\mD)\right)D_i\right]\right\|^2
\leq\ &
(1+\delta)^2
\left(\frac{70J(\mathsf{K}_0)}{T_J}
\max\left\{n\beta_0^2
,\frac{N}{1-\rho_W}\right\}\right)^2 \\
\leq\ &
\left(\frac{75J(\mathsf{K}_0)}{T_J}
\max\left\{n\beta_0^2
,\frac{N}{1-\rho_W}\right\}\right)^2.
\end{aligned}
$$

Next, by Lemma~\ref{lemma:sampling_error} and $\mD^0\sim \text{Uni}(\mathbb S_{n_K})$, we  bound \eqref{equ: P2} by
\begin{align*}
\sum_{i=1}^N\left\|\E\!\left[
J(\mK+r\mD)\left(D_i-D_i^0\right)\right]\right\|^2
\leq\ &
\sup_{\mK\in\mathcal{Q}^1}J(\mK')^2
\cdot \sum_{i=1}^N
\mathbb{E}\left[\|D_i-D_i^0\|^2_F\right] \\
\leq\ &
(20J(\mK_0))^2\cdot \delta^2
\mathbb{E}\left[\|D_i^0\|^2_F\right]
=\delta^2(20J(\mK_0))^2.
\end{align*}
Further, by $\kappa_0=\sup_{\mK\in\mathcal{Q}^1}\|\nabla J(\mK)\|
$, we can bound \eqref{equ: P3} by
\begin{align*}
\sum_{i=1}^N\left\|\E\!\left[
\left(J(\mK+r\mD)
-J(\mathsf{K}+r\mathsf{D}^0)\right)D_i^0\right]
\right\|_F^2
\leq\ &
\sum_{i=1}^N
\E\!\left[
r^2\kappa_0^2\|\mD-\mD^0\|^2\|D_i^0\|_F^2\right] \\
=\ &
r^2\kappa_0^2\E\!\left[
\|\mD-\mD^0\|^2\right]
\leq
\delta^2 r^2\kappa_0^2
\leq \delta^2 (20J(\mK_0))^2,
\end{align*}
where we used the assumption that $r\leq 20J(\mK_0)/\kappa_0$.

Now we summarize all the previous results and obtain
\begin{align*}
& \left\|\E\!\left[\hat{\mathsf{G}}^r(\mK,\mathsf{D})\right]- \nabla J(\mK)
\right\|^2 \\
\leq\ &
5\phi_0^2r^2
+\frac{5}{4}\frac{n_K^2}{r^2}
\Bigg[5\delta^2(20J(\mK_0))^2
+5\delta^2 (20J(\mK_0))^2
+\frac{5}{3}
\left(\frac{75J(\mathsf{K}_0)}{T_J}
\max\left\{n\beta_0^2
,\frac{N}{1-\rho_W}\right\}\right)^2\Bigg] \\
\leq\ &
5\phi_0^2r^2
+2\left(\frac{50\delta n_K J(\mK_0)}{r}\right)^2
+
5\left(\frac{50n_K J(\mathsf{K}_0)}{r T_J}
\max\left\{n\beta_0^2
,\frac{N}{1-\rho_W}\right\}\right)^2.
\end{align*}

\item {\it Proof of \eqref{equ: grad 2 moment bdd}:}
Let $\mK\in\mathcal{Q}^0$ be arbitrary. It can be seen that
\begin{align*}
\mathbb{E}
\!\left[
\big\|\hat{\mathsf{G}}^r(\mathsf{K},\mathsf{D})\big\|^2\right]
=\ &
\sum_{i=1}^N
\mathbb{E}
\!\left[
\big\|\hat{G}^r_i(\mathsf{K},\mathsf{D})\big\|^2_F\right]
=
\sum_{i=1}^N
\E\!\left[\frac{n_K^2}{r^2}\,
\E\!\left[
\left.\hat J_i(\mK \!+\! r\mD)^2\right| \mD\right]\cdot \|D_i \|_F^2 \right]\\
\leq\,&
\frac{n_K^2}{r^2} \max_{1\leq i \leq N}
\sup_{\mK' \in \mathcal{Q}^1}
\!\!
\left\{\E_w\!\big[\hat{J}_i(\mK')^2\big]
\right\}
\cdot
\mathbb{E}\!\left[\sum_{i=1}^N
\|D_i\|_F^2
\right]\\
\leq \,&
\frac{n_K^2}{r^2}(1+\delta)^2 \max_{1\leq i \leq N}
\sup_{\mK' \in \mathcal{Q}^1}
\!\!\left\{\E_w\!\big[\tilde{J}_i(\mK')^2\big]\right\},
\end{align*}
where the last inequality follows from $0\leq \hat J_i(\mathsf{K}')
\leq \tilde J_i(\mathsf{K}')$. On the other hand, for any $\mathsf{K}'\in\mathcal{Q}^1$, we have
$$
\begin{aligned}
\mathbb{E}_w\!\big[\tilde J_i(\mathsf{K}')^2\big]
\leq\,&
4\,
\E_w\!\left[ \big(\tilde J_i(\mK')-J(\mK')\big)^2
\right]+  \frac{4}{3} J(\mK') ^2\\
\leq\,&
4\left[
\frac{6nJ(\mathsf{K}')^2}{T_J}\beta_0^2
+
\frac{8J(\mathsf{K}')^2}{T_J^2}\left(\frac{N}{1 \!-\! \rho_W}\right)^2
\right] 
+ \frac{4}{3}J(\mK')^2\\
\leq\,&
4(20J(\mK_0))^2\Bigg[\frac{6n\beta_0^2}{T_J}
+
\frac{8}{T_J^2}\left(\frac{N}{1 \!-\! \rho_W}\right)^2\Bigg]
+ \frac{4}{3}(20 J(\mK_0))^2 \\
\leq\,&
(20J(\mathsf{K}_0))^2
\left[
4\left(\frac{6}{120}
+\frac{8}{120^2}\right)+\frac{4}{3}
\right]
< (28J(\mathsf{K}_0))^2,
\end{aligned}
$$
where the second inequality uses \eqref{eq:square_diff_mu_i_J} in Lemma \ref{lem: bias and var of mui}, the third inequality follows from $J(\mathsf{K}')\leq 20J(\mathsf{K}_0)$ for $\mathsf{K}'\in\mathcal{Q}^1$, and the last two inequalities follow from the condition on $T_J$. By combining this bound with previous results and noting that $1\!+\!\delta\leq 15/14$, we obtain the bound on $\mathbb{E}\!\left[\big\|\hat{\mathsf{G}}^r(\mathsf{K},\mathsf{D})\big\|^2\right]$.
\end{itemize}
\end{proof}

\subsection{Analysis of One-Step Stochastic Gradient Update}
Step~4 of Algorithm~\ref{alg:main} can be viewed as a stochastic gradient descent update with biased gradient estimation. In this part, we characterize the change in the objective value of this step.

We shall use $\mathcal{F}_s$ to denote the filtration $\sigma(K_i(s'):s'\leq s)$ for each $s=1,\ldots,T_G$.


\begin{lemma}\label{lem: key lemma}
 Suppose $\bar J\geq 50J(\mK_0)$ and
\begin{align*}
&r\leq
\min\left\{\frac{14}{15}\xi_0,
\frac{20J(\mK_0)}{\kappa_0}\right\},
\ 
\eta \leq
\min\left\{\frac{14\xi_0 r}{15n_K\bar J},
\frac{1}{25\phi_0}
\right\},\\
&T_J \geq
120\max\left\{
n\beta_0^2,
\frac{N}{1-\rho_W}\right\},
\quad
T_S \geq\frac{\log(8N^2/(\phi_0\eta))}{-2\log\rho_W}.
\end{align*}
Then, as long as $\mK(s)\in \mathcal{Q}^0$, we will have $\mK(s+1)\in \mathcal Q^1$ and
\begin{equation}\label{eq:one_step_policy_grad_bound}
\E[J(\mK(s+1))\mid \mathcal{F}_s]
\leq J(\mK(s))- \frac{\eta}{2}\|\nabla J(\mK(s))\|^2 +
\frac{\eta}{2}Z
\end{equation}
where
\begin{align}
Z
\coloneqq\ &
5
\!\left[\phi_0^2 r^2 + \left(\frac{50J(\mathsf{K}_0)n_K}{rT_J}
\max\left\{
n\beta_0^2,
\frac{N}{1 \!-\! \rho_W}\right\}\right)^{\!2}
\right]
+
\phi_0 \eta
\left(
40J(\mathsf{K}_0)\frac{n_K}{r}
\right)^2.\label{equ: Z's formula}
\end{align}

\end{lemma}
\begin{proof}
By denoting $\delta = \sqrt{\phi_0\eta/8}$, we see that $N\rho_W^{T_S}\leq\delta$, and by the upper bound on $\eta$ we have $\delta\leq 1/14$, so $\delta$ satisfies the condition in Lemma \ref{lem:bias_var_grad_est}. Therefore, by \eqref{equ: |D|<1+delta}, we have
\begin{align*}
\|\mK(s\!+\!1)\!-\!\mK(s)\|^2
=\ &
\eta^2\sum_{i=1}^N \|\hat{G}^r_i(s)\|_F^2
\leq
\eta^2 \sum_{i=1}^N 
\frac{n_K^2\bar J^2 }{r^2}\|  D_i(s)\|_F^2 \\
\leq\ &
\left(\frac{(1+\delta)\eta\,n_K \bar J}{r}\right)^2\leq \xi_0^2,
\end{align*}
which implies $\mK(s+1)\in \mathcal Q^1$ as long as $\mK(s) \in \Q^0$. Secondly, since $\mK(s)\in \Q^0$, by  Lemma \ref{lemma:J_smoothness}, we have
\begin{align*}
J(\mK(s\!+\!1)) 
\leq\!
J(\mK(s)) \!-\! \eta\langle \nabla\! J(\mK(s)), \hat{\mathsf{G}}^r\!(s)\rangle \!+\!
\frac{\phi_{0}\eta^2\!}{2}\big\| \hat \mG^r\!(s)\big\|^2.
\end{align*}
Taking expectation conditioned on the filtration $\mathcal{F}_s$ yields
$$
\begin{aligned}
\E [J(\mK(s+1)) \,|\, \mathcal{F}_s ]
\leq \ &
J(\mK(s)) -\eta \left\langle \nabla J(\mK(s)), \E\!\left[\!\left.\hat \mG^r(s) \right| \mathcal{F}_s\right] \right\rangle
+ \frac{\phi_0}{2} \eta^2 
\E\!\left[
\!\left.\big\| \hat \mG^r(s)\big\|^2\right| \mathcal{F}_s\right] \\
=\ &
J(\mK(s)) -\eta \|\nabla J(\mK(s))\|^2
+ \frac{\phi_0}{2} \eta^2 \E\!\left[\!\left.\big\| \hat \mG^r(s)\big\|^2\right| \mathcal{F}_s\right] \\
&
+\eta  \left\langle \nabla J(\mK(s)), \nabla J(\mK(s))-\E\!\left[\!\left.\hat \mG^r(s)\right| \mathcal{F}_s\right] \right\rangle \\
\leq\ & J(\mK(s)) -  \frac{\eta}{2}  \|\nabla J(\mK(s))\|^2
+ \frac{\phi_0}{2} \eta^2\E\!\left[\!\left.\big\| \hat \mG^r(s)\big\|^2\right| \mathcal{F}_s\right] \\
&
+ \frac{\eta}{2}
\left\| \nabla J(\mK(s))-\E\!\left[\!\left.\hat \mG^r(s)\right| \mathcal{F}_s\right]\right\|^2,
\end{aligned}
$$
where we used Cauchy's inequality in the last step. 

By applying the 
results of Lemma \ref{lem:bias_var_grad_est} to  above, we obtain
$$
\begin{aligned}
\E [J(\mK(s+1)) \,|\, \mathcal{F}_s ]
\leq\ &
J(\mK(s))
-\frac{\eta}{2}  \|\nabla J(\mK(s))\|^2 \\
&
+\frac{\phi_0\eta^2}{2}
\left(30J(\mK_0)\frac{n_K}{r}\right)^2
+
\frac{\eta}{2}
\cdot 2\left(\frac{50\delta n_K J(\mK_0)}{r}\right)^2
\\
&
+\frac{\eta}{2}
\cdot 5\left[\phi_0^2 r^2+\left(\frac{50n_K J(\mathsf{K}_0)}{r T_J}
\max\left\{n\beta_0^2
,\frac{N}{1-\rho_W}\right\}\right)^2\right] \\
\leq\ &
J(\mK(s))
-\frac{\eta}{2}  \|\nabla J(\mK(s))\|^2
+\frac{\phi_0\eta^2}{2}
\left(40J(\mK_0)\frac{n_K}{r}\right)^2
\\
&
+\frac{\eta}{2}
\cdot 5\left[\phi_0^2r^2+\left(\frac{50n_K J(\mathsf{K}_0)}{r T_J}
\max\left\{n\beta_0^2
,\frac{N}{1-\rho_W}\right\}\right)^2\right],
\end{aligned}
$$
which concludes the proof.
\end{proof}


\subsection{Proving Stability of the Output Controllers}
\label{subsec:proof_stability}
Next, we show that all the output controllers $\{\mK(s)\}_{s=1}^{T_G}$ are in $\Q^0$ with high probability, which then implies that all the output controllers are stabilizing with high probability.

We assume that the algorithmic parameters satisfy the conditions in Theorem~\ref{theorem:main}. It's not hard to see that for sufficiently small $\epsilon>0$, the conditions of Lemma~\ref{lem: key lemma} are satisfied.

We define a stopping time $\tau$ to be the first time step when $K(s)$ escapes $\Q^0$:
\begin{equation}\label{equ: stopping time def tau}
\tau\!
\coloneqq
\min\left\{s\in\{1,\ldots,T_G \!+\! 1\}: J(\mK(s))>10 J(\mK_0)\right\}.
\end{equation}
Our goal is then to bound the probability $\Pb(\tau \leq T_G)$.
We first note that, under the conditions of Theorem~\ref{theorem:main},
\begin{align}
Z\leq\,&
5 \!\left[\phi_0^2 r^2
	\!+\!
	\left(\!
	\frac{J(\mathsf{K}_0)n_K}{r}\frac{r\sqrt{\epsilon}}{20 J(\mathsf{K}_0)n_K}
	\!\right)^{\!2}\right]
	+
\phi_0
	\!\left(40J(\mathsf{K}_0)\frac{n_K}{r}\right)^2\!\!
	\frac{3\epsilon r^2}{320\phi_0
		\!\left(40J(\mathsf{K}_0)\right)^2\! n_K^2}\nonumber
\\
	\leq\,&
	5\left(
	\frac{\epsilon}{1600}
	+\frac{\epsilon}{400}\right)
	+\frac{3\epsilon}{320}
	=\frac{\epsilon}{40}.\label{equ: Z's bdd}
\end{align}
Now,
we define a nonnegative supermartingale $Y(s)$ by
$$
	Y(s) \coloneqq
	J(\mK(\min\{s,\tau\})) +(T_G \!-\! s)\cdot\frac{\eta}{2} Z,
	\quad 1 \leq s \leq T_G.
$$
It is straightforward  that $Y(s) \geq 0$ for $1\leq s \leq T_G$. To verify that it is a supermartingale, we notice that when $\tau>s$,
$$
\begin{aligned}
	\E [Y(s+1)| \mathcal{F}_s]
	=\ &
	\E[J(\mathsf{K}(s+1)) | \mathcal{F}_s ]+ (T_G \!-\! s \!-\! 1)\cdot\frac{\eta}{2} Z \\
	\leq\ &
	J(\mathsf{K}(s))- \frac{\eta}{2}\|\nabla J(\mathsf{K}(s))\|^2 + \frac{\eta}{2}Z
	+ (T_G\!-\!1\!-\!s) \cdot\frac{\eta}{2} Z \leq Y(s),
\end{aligned}
$$
and when $\tau\leq s$, 
$
	\E [Y(s+1)| \mathcal{F}_s]
	=J(\mathsf{K}(\tau) + (T_G\!-\!1\!-\!s) \frac{\eta}{2} Z 
	\leq Y(s).$
Now, by the monotonicity and Doob's maximal inequality for supermartingales, we obtain the following bound:
	\begin{align}
	\mathbb{P}(\tau \leq T_G)
	& \leq
	\mathbb{P}
	\left(\max_{s=1, \dots, T_G} Y(s)> 10 J(\mK_0)\right)
	\leq \frac{\E[Y(1)]}{10 J(\mK_0)} \nonumber\\
	& = \frac{J(\mK_0)+ (T_G-1)\eta Z/2}{10 J(\mK_0)}\leq \frac{1}{10}+ \frac{c}{20},
	\label{eq:prob_tau_good}
	\end{align}
where the last inequality used $T_G=c\cdot 40 J(\mK_0)/(\eta \epsilon)$ and $Z\leq \epsilon/40$. This implies that all the output controllers are stabilizing with probability at least $1-(1/10+ c/20)=9/10-c/20$.

\subsection{Proving the Performance Bound in Theorem \ref{theorem:main}.}
\label{subsec:proof_final_part}
To prove the performance bound \eqref{equ: main_thm_convergence}, we first extend the results in Lemma~\ref{lem: key lemma} and show that
\begin{equation}\label{equ: recursive for Delta_s}
\begin{aligned}
\E\!\left[ J(\mathsf{K}(s+1))\mathsf{1}_{\{\tau >s+1\}} | \mathcal{F}_s\right]
\leq\ &
J(\mathsf{K}(s))\mathsf{1}_{\{\tau >s\}} - \frac{\eta}{2} \| \nabla J(\mathsf{K}(s))\|^2 \mathsf{1}_{\{\tau >s\}} + Z,
\end{aligned}
\end{equation}
If  $\{\tau >s\}$ occurs, then $\mK(s)\in Q^0$, and it can be verified that the other conditions of Lemma \ref{lem: key lemma} also hold. Then, we have
\begin{align*}
\E[ J(\mathsf{K}(s+1))\mathsf{1}_{\{\tau >s+1\}} | \mathcal{F}_s]
\leq\ &
\E[ J(\mathsf{K}(s+1)) | \mathcal{F}_s] \\
\leq\ &
J(\mK(s)) - \frac{\eta}{2} \|\nabla J(\mK(s))\|^2 + \frac{\eta}{2}Z\\
=\ &
J(\mK(s)) \mathsf{1}_{\{\tau>s\}} - \frac{\eta}{2} \|\nabla J(\mK(s))\|^2 \mathsf{1}_{\{\tau>s\}}+ \frac{\eta}{2}Z.
\end{align*}
Otherwise, if $\{\tau \leq s\}$ occurs, by $Z\geq 0$, we trivially have
$ \E[ J(\mathsf{K}(s+1))\mathsf{1}_{\{\tau >s+1\}} | \mathcal{F}_s]
=0 
\leq
J(\mathsf{K}(s)) \mathsf{1}_{\{\tau>s\}}
- \frac{\eta}{2} \|\nabla J(\mK(s))\|^2 \mathsf{1}_{\{\tau>s\}}+ \frac{\eta}{2}Z$.
Summarizing the two cases, we have proved \eqref{equ: recursive for Delta_s}.

We then establish the following bound:
\begin{equation}\label{eq:exp_avg_grad_tau_good}
\E
\left[
\left(\frac{1}{T_G} \sum_{s=1}^{T_G} \| \nabla J(\mK(s))\|^2\right)\mathsf{1}_{\{ \tau > T_G  \}}\right]\leq \frac{\epsilon}{40}+ \frac{\epsilon}{20c}.
\end{equation}
\begin{proof}[{Proof of \eqref{eq:exp_avg_grad_tau_good}}]
By taking the total expectation of \eqref{equ: recursive for Delta_s}, we have 
\begin{align*}
\E\!
\left[J(\mathsf{K}(s+1)) \mathsf{1}_{\{ \tau >s+1\}}\right]
\leq
\E\!\left[J(\mathsf{K}(s)) \mathsf{1}_{\{ \tau >s\}}
\right]
- \frac{\eta}{2}\E\!\left[\|\nabla J(\mK(s))\|^2  \mathsf{1}_{\{ \tau >s\}}\right] +
\frac{\eta}{2}Z.
\end{align*}
By reorganizing terms and taking the telescoping sum, we have
\begin{align*}
\frac{1}{T_G}\sum_{s=1}^{T_G}\E
\!\left[
\|\nabla J(\mK(s))\|^2 \mathsf{1}_{\{ \tau >T_G\}}\right]
\leq\,&
\frac{1}{T_G}\sum_{s=1}^{T_G}
\E\!\left[
\|\nabla J(\mK(s))\|^2 \mathsf{1}_{\{ \tau >s\}}\right] \\
\leq\,
&\frac{2}{\eta T_G}
\mathbb{E}\!\left[J(\mathsf{K}(1))
-J(\mathsf{K}(T_G+1))\one_{(\tau>T_G+1)}
\right]+Z \\
\leq\ &
\frac{2}{\eta}\cdot\frac{\eta\epsilon}{ 40cJ(\mathsf{K}_0)}J(\mathsf{K}_0) + \frac{\epsilon}{40}
\leq \frac{\epsilon}{20c}+  \frac{\epsilon}{40},
\end{align*}
where we used the fact that $J(\mK)\geq 0$ over $\mK\in\mathcal{K}_{\mathrm{st}}$ and \eqref{equ: Z's bdd}.
\end{proof}

By \eqref{eq:exp_avg_grad_tau_good} above and the bound \eqref{eq:prob_tau_good}, the performance bound \eqref{equ: main_thm_convergence} of Theorem~\ref{theorem:main} can now be proved as follows:
\begin{align*}
\Pb\left(\frac{1}{T_G} \sum_{s=1}^{T_G}\|\nabla J(\mK(s))\|^2 \geq \epsilon\right)
=\, &
\Pb\left(\frac{1}{T_G} \sum_{s=1}^{T_G}\|\nabla J(\mK(s))\|^2 \geq \epsilon, \tau > T_G\right) \\
& +
\Pb
\left(\frac{1}{T_G} \sum_{s=1}^{T_G}\|\nabla J(\mK(s))\|^2 \geq \epsilon, \tau \leq T_G\right)\\
\leq\,&
\Pb\left(\frac{1}{T_G} \sum_{s=1}^{T_G}\|\nabla J(\mK(s))\|^2 \mathsf{1}_{\{ \tau > T_G\}} \geq \epsilon\right)
+
\Pb\left(\tau \leq T_G\right)\\
\leq\,&
\frac{1}{\epsilon}\,\E\!
\left[\frac{1}{T_G} \sum_{s=1}^{T_G}\|\nabla J(\mK(s))\|^2 \mathsf{1}_{\{ \tau > T_G\}}\right] +\Pb\left(\tau \leq T_G\right) \\
\leq\ &
\frac{1}{8}+ \frac{1}{20c}+ \frac{c}{20},
\end{align*}
where we used Markov's inequality. We also have
\begin{align*}
\Pb\left(\|\nabla J(\hat\mK)\|^2 \geq \epsilon\right)
\leq\,&
\Pb\left(\|\nabla J(\hat\mK)\|^2 \mathsf{1}_{\{ \tau > T_G\}} \geq \epsilon\right)
+
\Pb\left(\tau \leq T_G\right)\\
\leq\,&
\frac{1}{\epsilon}\,\E\!
\left[ \|\nabla J(\hat\mK)\|^2 \mathsf{1}_{\{ \tau > T_G\}}\right] +\Pb\left(\tau \leq T_G\right) \\
= \ & \frac{1}{\epsilon}\,\E \left[\E\!
\left[ \|\nabla J(\hat\mK)\|^2 \mathsf{1}_{\{ \tau > T_G\}}\mid \{\mK(s)\}_{s=1}^{T_G}\right]\right] +\Pb\left(\tau \leq T_G\right)\\
=\ &
\frac{1}{\epsilon}\,\E\!
\left[\frac{1}{T_G} \sum_{s=1}^{T_G}\|\nabla J(\mK(s))\|^2 \mathsf{1}_{\{ \tau > T_G\}}\right] +\Pb\left(\tau \leq T_G\right) \\
\leq\ &
\frac{1}{8}+ \frac{1}{20c}+ \frac{c}{20}.
\end{align*}
where the second equality follows by noticing that, conditioning on $ \{\mK(s)\}_{s=1}^{T_G}$, $\mathsf{1}_{\{ \tau > T_G\}}$ is a constant and $\hat\mK$ is uniformly randomly selected from $\{\mK(s)\}_{s=1}^{T_G}$.
\section{{Numerical Studies}}\label{sec:simulation}



In this section, we numerically test our ZODPO on  Heating Ventilation and Air
Conditioning (HVAC) systems for  multi-zone buildings. We consider both time-invariant cases as theoretically analyzed above and the time-varying cases for more realistic implementation. We will focus on ZODPO  since we are not aware of  other learning algorithms in literature that can be directly applied to our problem.

\begin{figure}[ht]
\centering

\begin{subfigure}{0.40\textwidth}
\centering
\includegraphics[width =1 \textwidth]{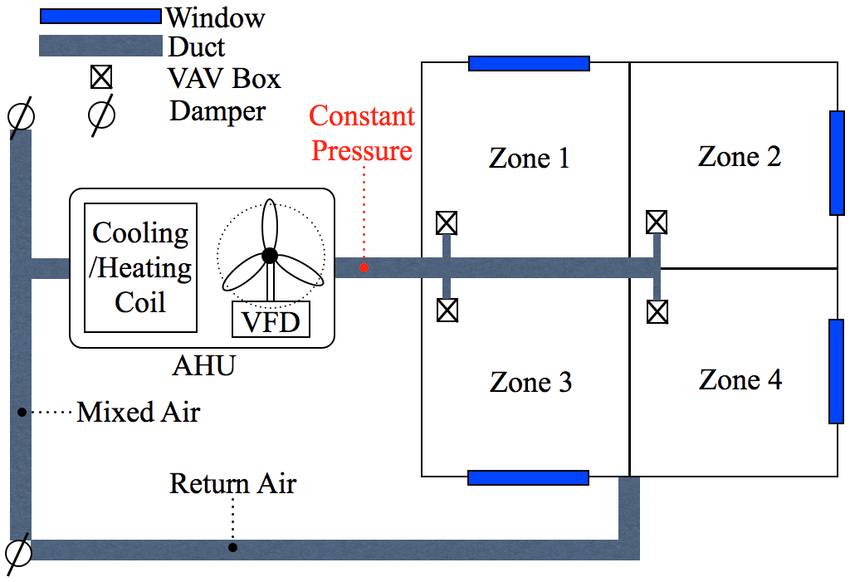}
\caption{A 4-zone HVAC system}
\end{subfigure}
\\
\begin{subfigure}{0.32\textwidth}
\centering
\includegraphics[width =  \textwidth]{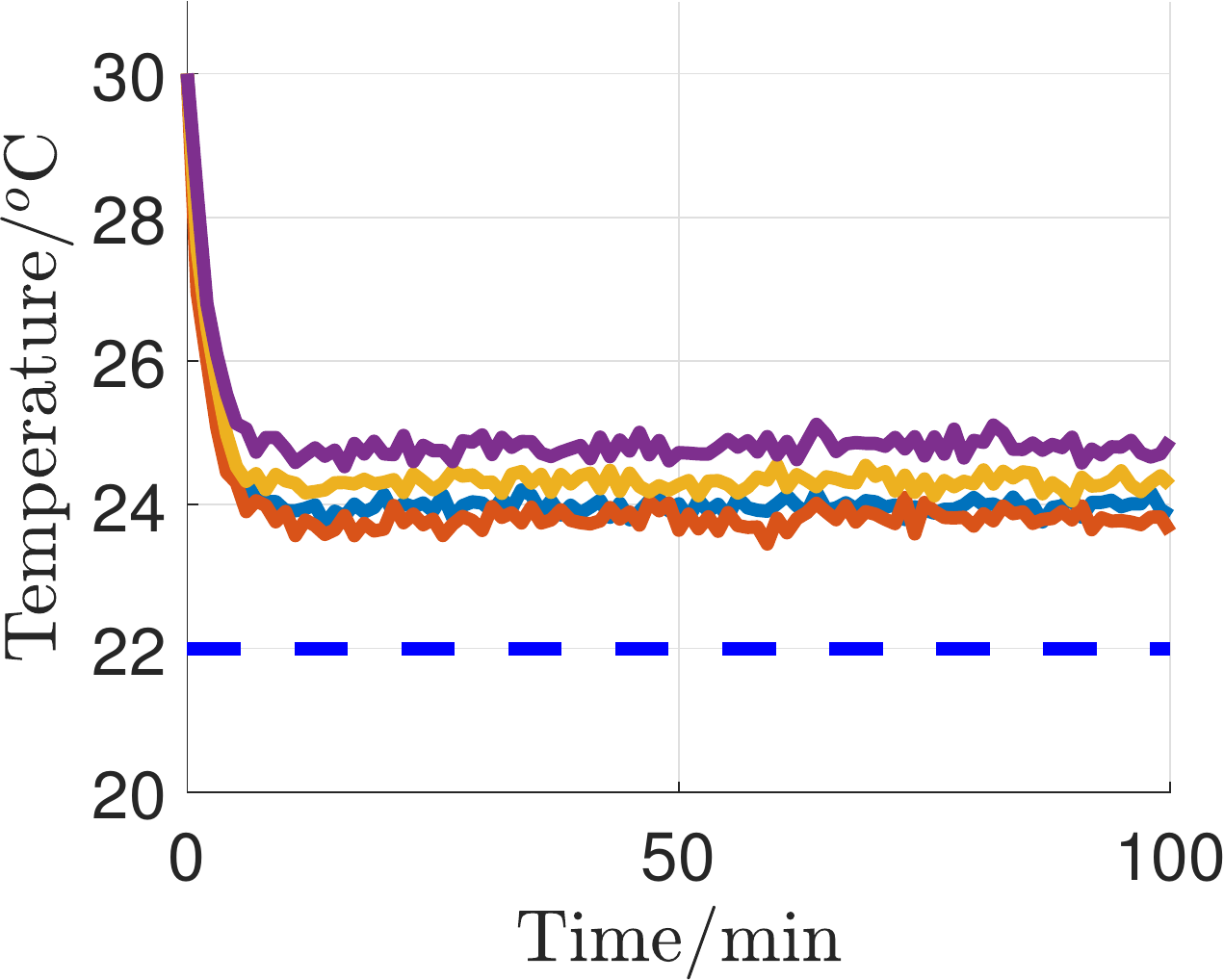}
\caption{$T_G\!=\!50$}
\end{subfigure}
~
\begin{subfigure}{0.32\textwidth}
\centering
\includegraphics[width = \textwidth]{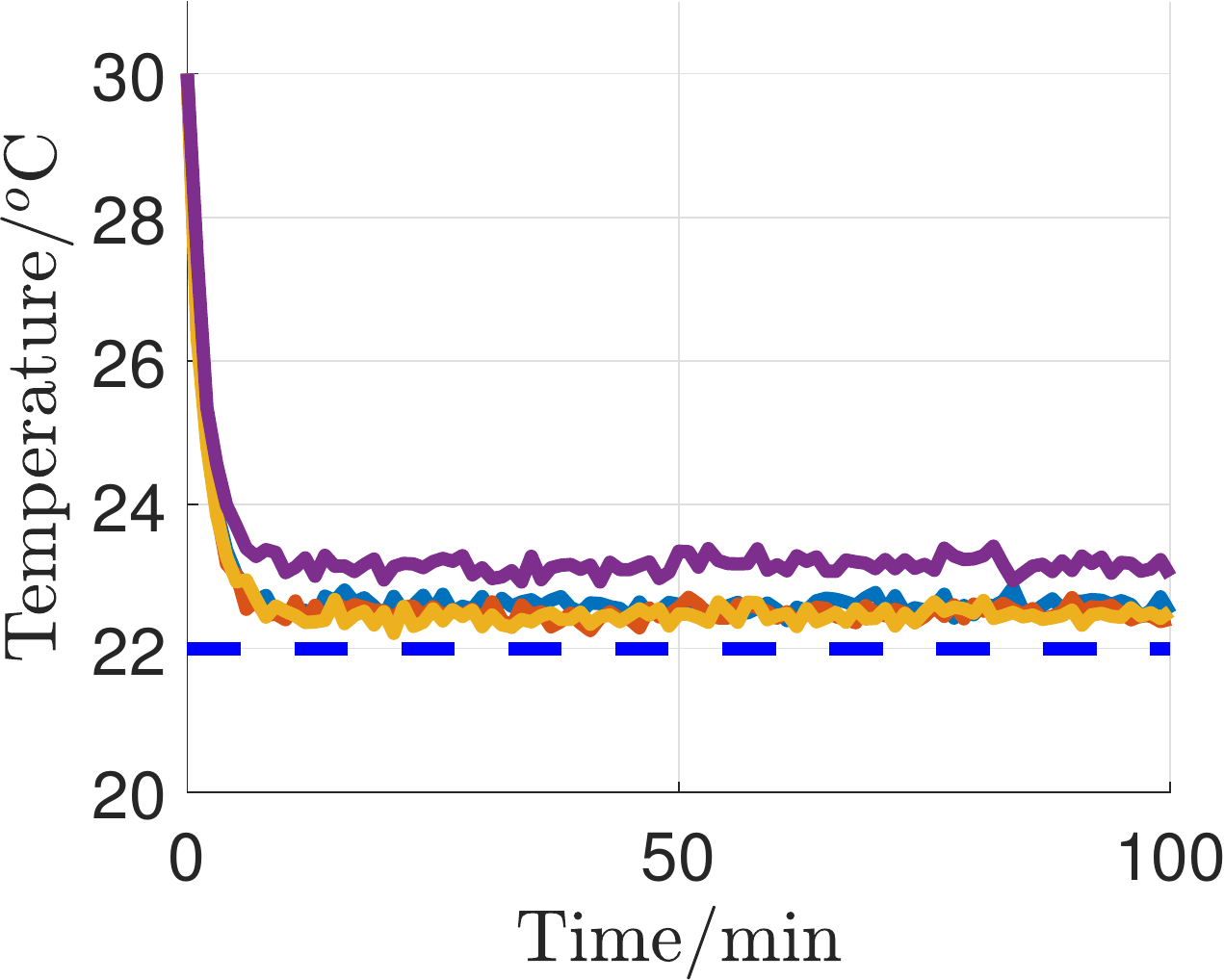}
\caption{$T_G\!=\!150$}
\end{subfigure}
~
\begin{subfigure}{0.32\textwidth}
\centering
\includegraphics[width = \textwidth]{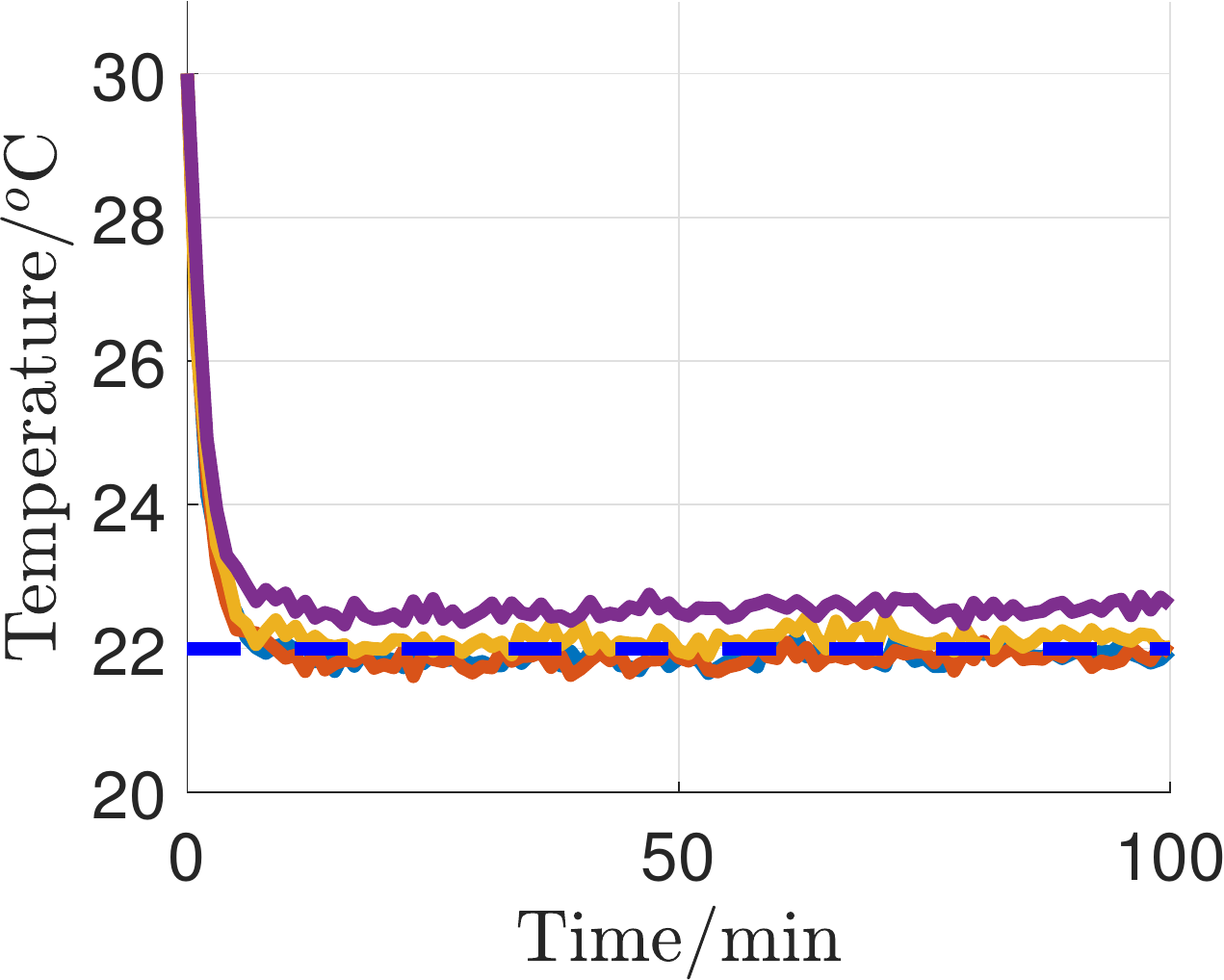}
\caption{$T_G\!=\!250$}
\end{subfigure}
\caption{(a) is a diagram of the 4-zone HVAC system considered in Section \ref{subsec:time invariant}. The figure is from \cite{zhang2016decentralized}. (b)-(d) shows the dynamics of indoor temperatures of the 4 zones under the controllers generated by ZODPO after $T_G=50, 150, 250$ iterations.}
\label{fig:simulation_results}
\end{figure}

\begin{figure}[ht]
    \centering
    \includegraphics[width =0.50 \textwidth]{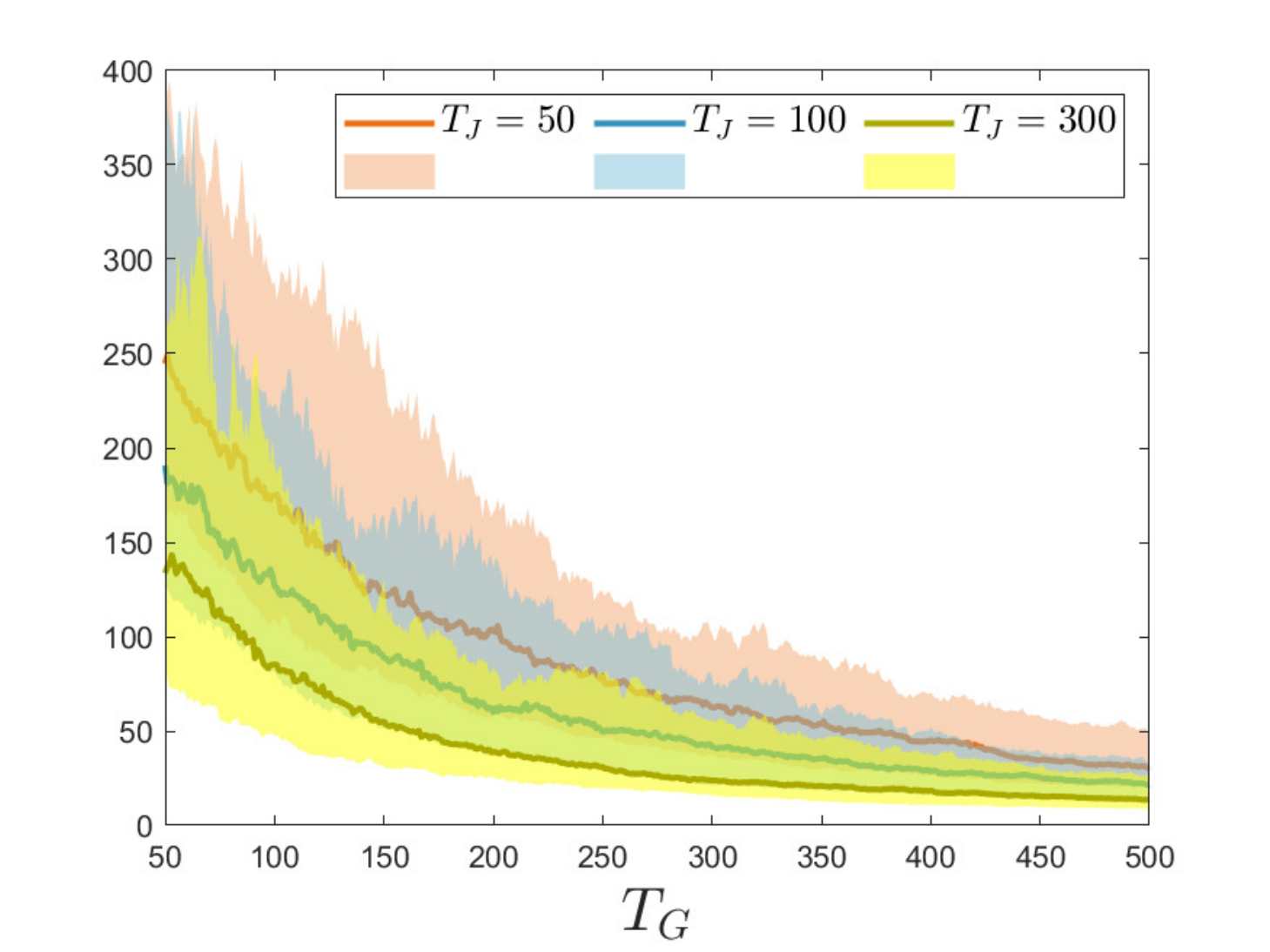}
    \caption{A comparison of ZODPO with $T_J=50, 150, 300$. The solid lines represent the mean values and the shade represents 70\% confidence intervals of the actual costs by implementing the controllers generated by ZODPO. }
    \label{fig:TI conf interval}
\end{figure}

\subsection{Thermal Dynamics Model}
This paper considers   multi-zone buildings with HVAC systems. Each zone is equipped with a sensor that can  measure the local temperatures, and can adjust the supply air flow rate of its associated HVAC system.

We adopt the linear thermal dynamics model studied in \cite{zhang2016decentralized} with additional process noises in the discrete time setting, i.e. 
\begin{align*}
    x_i(t\!+\!1)\!-\!x_i(t)\!=& \frac{\Delta }{\upsilon_i\zeta_i}(\theta^{\mathrm{o}}(t)\!-\!x_i(t))\! +\! \sum_{j=1}^N\! \frac{\Delta}{\upsilon_i \zeta_{ij}}\!(x_j(t)\!-\!x_i(t))\\
    &+ \!\frac{\Delta}{\upsilon_i} u_i(t)\! +\! \frac{\Delta}{\upsilon_i} \pi_i \!+\! \frac{\sqrt{\Delta}}{\upsilon_i} w_i(t),  \ \ \ 1 \!\leq\! i \!\leq\! N,
\end{align*}
where $x_i(t)$ denotes the temperature of zone $i$ at time $t$, $u_i(t)$ denotes the control input of zone $i$ that is related with the air flow rate of the HVAC system, $\theta^{\mathrm{o}}(t)$ denotes the outdoor temperature, $\pi_i$ represents a constant heat from external sources to zone $i$,  $w_i(t)$ represents random disturbances, $\Delta$ is the time resolution, $\upsilon_i$ is the thermal capacitance of zone $i$, $\zeta_i$ represents the thermal resistance of the windows and walls between the zone $i$ and outside environment, and $\zeta_{ij}$ represents the thermal resistance of the walls between zone $i$ and $j$. 

At each zone $i$, there is a desired temperature $\theta_i^*$ set by the users. The local cost function  is composed by the deviation from the desired temperature and the control cost, i.e.
$c_i(t)=\big(x_{i}(t) - \theta_i^*\big)^2 + \alpha_i u_{i}(t)^2
$, where $\alpha_i>0$ is a trade-off parameter.




\subsection{Time-Invariant Cases}\label{subsec:time invariant}





In this subsection, we consider a system with $N=4$ zones (see Figure \ref{fig:simulation_results}(a)) and a time-invariant outdoor temperature $\theta^{\mathrm{o}}=30\,{}^\circ\mathrm{C}$. The system parameters are listed below.  We set $\theta_i^*=22\,{}^\circ\mathrm{C}$, $\alpha_i=0.01$, $\pi_i=1\mathrm{kW}$, $\Delta=60$s, $\upsilon_i = 200 \mathrm{kJ}/{}^\circ\mathrm{C}, \zeta_i=1 {}^\circ\mathrm{C}/\mathrm{kW}$ for all $i$, $\zeta_{ij}=1 {}^\circ\mathrm{C}/\mathrm{kW}$ if zone $i$ and $j$ have common walls and $\zeta_{ij}=0 {}^\circ\mathrm{C}/\mathrm{kW}$ otherwise.  Besides, we consider i.i.d. $w_i(t)$  following $N(0, 2.5^2)$.

We consider the following  decentralized control policies:
$$u_i(t)=K_i x_i(t)+b_i, \quad \forall 1\leq i \leq N$$
where a constant term $b_i$ is adopted to deal with nonzero desired temperature $\theta_i^*$ and the constant drifting term in the system dynamics $\pi_i$. We apply ZODPO to learn both $K_i$ and $b_i$.\footnote{This requires a  straightforward modification of Algorithm \ref{alg:main}: in Step 1, add perturbations onto both $K_i$ and $b_i$, in Step 2,   estimate the partial gradients with respect to $K_i$ and $b_i$, in Step 3, update $K_i$ and $b_i$ by the estimated partial gradient.} The algorithm parameters are listed below. We consider the  communication network $W$ where $W_{ii}=1/2$, and $W_{ij}=W_{ji}=1/4$ if $i\not=j$ and $i$ are $j$ share common walls. We set $r=0.5, \eta=0.0001, \bar J =10^6$. Since the thermal dynamical system is open-loop stable, we select the initial controller as zero, i.e. $\mK_i=0, b_i=0$ for all $i$.

Figures \ref{fig:simulation_results}(b)--(d) plot the temperature dynamics of the four zones  by implementing the controllers generated by ZODPO at policy gradient iterations $T_G=50,150,250$ respectively with $T_J=300$. It can be observed that with more iterations, the controllers generated by ZODPO stabilize the system faster and steer the room temperature closer to the desired temperature.


Figure \ref{fig:TI conf interval} plots the infinite-horizon averaged costs of  controllers  generated by ZODPO for different $50 \leq T_G \leq 500$ when $T_J=50,100, 300$ by 500 repeated simulations. 
As $T_G$ increases, the averaged costs keep decreasing, which is consistent with Figures \ref{fig:simulation_results}(b)--(d). Since $T_G\leq 300$ is not extremely large, we do not observe the increase of the probabilities of generating unstable controllers. Notice that with a larger $T_J$, the confidence intervals shrink and the averaged costs decrease, indicating less fluctuations and better performance. This is intuitive since a larger $T_J$ indicates a better gradient estimation. 



\subsection{Larger Scale Systems}
Here, we consider an $N=20$ system  to demonstrate that our ZODPO can handle systems with higher dimensions. We consider a 2-floor building with $5\times 2$ rooms on each floor.  Other system parameters are provided in Section \ref{subsec:time invariant}. 
Figure \ref{fig:TV}(a) plots the dynamics of the indoor temperatures of 20 rooms when implementing a controller generated by ZODPO after $T_G=4800$ iterations with $T_J=500$. Notice that all the room temperatures stabilize around the desired temperature 22${}^\circ \mathsf{C}$, indicating that our ZODPO can output an effective controller for reasonably large $T_G$ and $T_J$ even for a larger-scale system.







\begin{figure}
\centering
\begin{subfigure}{0.32\textwidth}
\centering
\includegraphics[width = \textwidth]{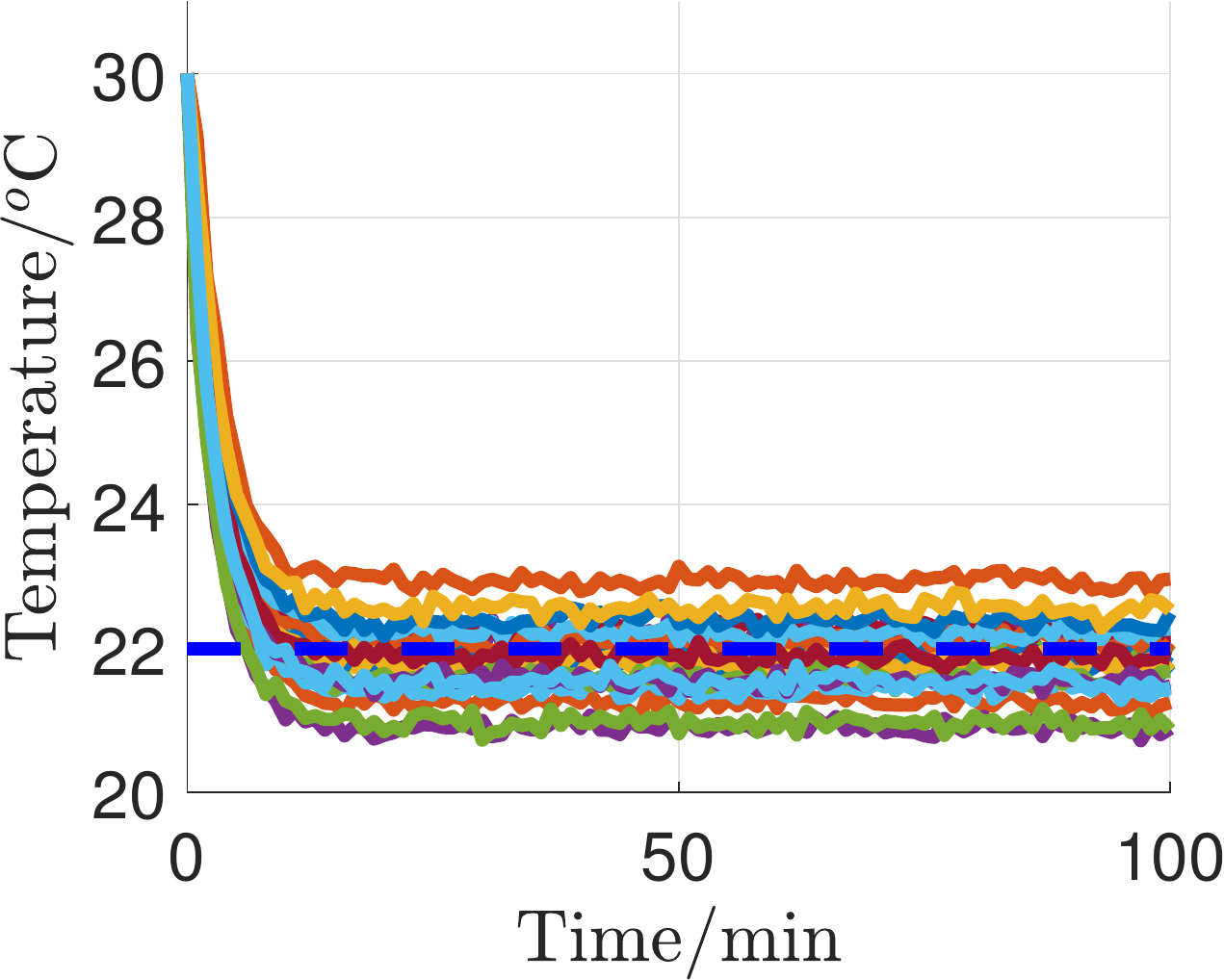}
\caption{An $N=20$ system}
\end{subfigure}
\hfil
\begin{subfigure}{0.32\textwidth}
\centering
\includegraphics[width = \textwidth]{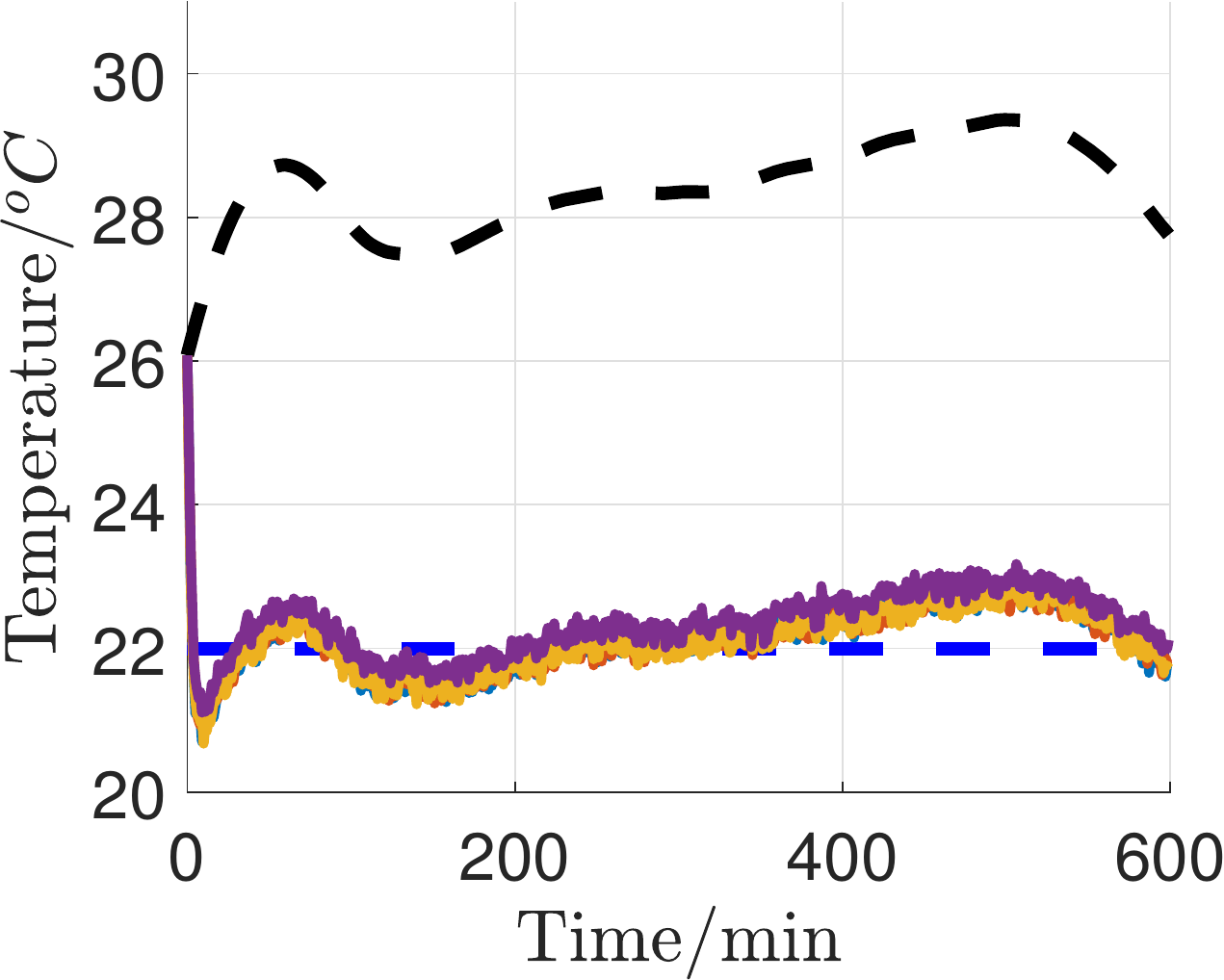}
\caption{Varying outdoor temperature}
\end{subfigure}
\setlength{\belowcaptionskip}{-10pt}
\caption{(a) plots the dynamics of indoor temperatures of an $N=20$ system with a constant outdoor temperature $30{}^\circ\mathrm{C}$. (b) plots the dynamics of indoor temperatures of the 4-zone system with time-varying outdoor temperature, with the black line representing the outdoor temperature. }
\label{fig:TV}
\end{figure}

\subsection{Varying Outdoor Temperature}

This subsection considers a more realistic scenario where the outdoor temperature is changing with time. The data is collected by Harvard HouseZero Program.\footnote{\url{http://harvardcgbc.org/research/housezero/}}
To adapt to varying outdoor temperature, we consider the following form of controllers:
$$u_i(t)=K_ix_i(t) +K_i^{\mathrm{o}}\theta^{\mathrm{o}}(t)+b_i,$$
and apply our ZODPO to learn $K_i, K_i^{\mathrm{o}}, b_i$. During the learning process, we consider different outdoor temperatures in different policy gradient iterations, but fix the outdoor temperature within one iteration for better training. We consider the system in Section \ref{subsec:time invariant} and set $T_J=300$, $\eta=2\times 10^{-6}, r=0.5, \bar J=10^8$. 
Figure \ref{fig:TV}(b) plots the temperature dynamics by implementing the controller generated by ZODPO at policy iteration $T_G=250$. The figure shows that even with varying outdoor temperatures, ZODPO is still able to find a controller that roughly maintains the room temperature at the desired level.

\section{{Conclusions and Future Work}}
This paper considers distributed learning of decentralized linear quadratic control systems with limited communication, partial observability, and local costs. We propose a ZODPO algorithm that allows agents to learn decentralized controllers in a distributed fashion by leveraging the ideas of policy gradient, consensus algorithms and zero-order optimization. We prove the stability of the output controllers with high probability. We also provide sample complexity guarantees. Finally, we numerically test ZODPO on HVAC systems.

There are various directions for future work. For example, effective initialization and exploration of other stabilizing components to approach the global optima are important topics. It is also worth exploring the optimization landscape of decentralized LQR systems with additional structures. Besides, we are also interested in employing other gradient estimators, e.g. two-point zero-order gradient estimator, to improve sample complexity. Utilizing more general controllers that take advantage of history information is also an interesting direction. Finally, it is worth designing actor-critic-type algorithms to reduce the running time per iteration of policy gradient.

\appendix

\section{Additional Notations and Auxiliary Results}

\nbf{Notations:}
Recall that $\mathcal{M}(\mathsf{K})\in\mathbb{R}^{n\times m}$ denotes the global control gain given $\mathsf{K}\in\mathcal{K}_{\mathrm{st}}$, and notice that $\mathcal{M}$ is an injective linear map from $\mathbb{R}^{n_K}$ to $\mathbb{R}^{n\times m}$.
For simplicity, we denote
$$
\begin{aligned}
&A_{\mathsf K} \coloneqq A+B\mathcal{M}(\mK),
\quad
Q_{i,\mathsf{K}} \coloneqq\  Q_i+\mathcal{M}(\mK)^\top  R_i \mathcal{M}(\mK), \\
&Q\coloneqq\frac{1}{N}\sum_{i=1}^N Q_i, \quad R\coloneqq\frac{1}{N}\sum_{i=1}^N R_i,\quad
Q_{\mathsf{K}}\coloneqq
\frac{1}{N}\sum_{i=1}^N Q_{i,\mathsf{K}}.
\end{aligned}
$$

Besides, we define  
\begin{align}\label{equ: Sigma(K,t) def}
\Sigma_{\mK,t}=\E\!\left[x(t)x(t)^\top\right], \quad \Sigma_{\mK,\infty}=\lim_{t\to \infty}\E\!\left[x(t)x(t)^\top\right],
\end{align}
where $x(t)$ is the state generated by controller $\mK$. Notice that
\begin{align}\label{equ: Sigma(K,t)}
\Sigma_{\mathsf{K},t}=  \sum_{\tau=0}^{t-1}A_{\mathsf{K}}^\tau\Sigma_w \big(A_{\mathsf{K}}^\top\big)^\tau
\preceq
\Sigma_{\mathsf{K},\infty}\!=\!  \sum_{\tau=0}^{\infty}A_{\mathsf{K}}^\tau\Sigma_w \big(A_{\mathsf{K}}^\top\big)^\tau.
\end{align} 
The objective function $J(\mathsf{K})$ can be represented as
\begin{equation}\label{equ: formula of J(K)}
J(\mathsf{K})=\lim_{t\to \infty}
\E\!\left[x(t)^\top Q_{\mK} x(t)\right]
=\operatorname{tr}(Q_{\mathsf{K}}\Sigma_{\mathsf{K},\infty}),
\end{equation}

\nbf{Auxiliary results:} We provide an auxiliary lemma showing that $\|\big(\Sigma_{\mathsf{K},\infty}^{-\frac{1}{2}} A_{\mathsf{K}}
\Sigma_{\mathsf{K},\infty}^{\frac{1}{2}}\big)^{\!t}\|$ decays exponentially as $t$ increases.
\begin{lemma}\label{lemma:bound_iteratedA}
There exists a continuous function $\varphi:\mathcal{K}_{\mathrm{st}}\rightarrow[1,+\infty)$ such that
$$
\left\|\left(\Sigma_{\mathsf{K},\infty}^{-\frac{1}{2}} A_{\mathsf{K}}
\Sigma_{\mathsf{K},\infty}^{\frac{1}{2}}\right)^t\right\|
\leq \varphi(\mathsf{K})\left(
\frac{1+\rho(A_{\mathsf{K}})}{2}\right)^t
$$
for any $t\in\mathbb{N}$ and any $\mathsf{K}\in\mathcal{K}_{\mathrm{st}}$.
\end{lemma}
\begin{proof}
Denote $\tilde{\rho}(A_\mathsf{K})\coloneqq(1+\rho(A_\mathsf{K}))/2$. Since $
\rho\left(\Sigma_{\mathsf{K},\infty}^{-\frac{1}{2}}A_\mathsf{K}\Sigma_{\mathsf{K},\infty}^{\frac{1}{2}}\right)=\rho(A_\mathsf{K})<\tilde{\rho}(A_\mathsf{K})
$, 
the matrix series
$$
\tilde P_\mathsf{K}=\sum_{t=0}^\infty \tilde{\rho}(A_\mathsf{K})^{-2t}
\!\left(\Sigma_{\mathsf{K},\infty}^{-\frac{1}{2}}A_\mathsf{K} \Sigma_{\mathsf{K},\infty}^{\frac{1}{2}}\right)^t\!\left(\Sigma_{\mathsf{K},\infty}^{\frac{1}{2}}A_\mathsf{K}^\top \Sigma_{\mathsf{K},\infty}^{-\frac{1}{2}}\right)^t
$$
converges, and satisfies the Lyapunov equation
$$
\tilde{\rho}(A_\mathsf{K})^{-2}\left(\Sigma_{\mathsf{K},\infty}^{-\frac{1}{2}}A_\mathsf{K}\Sigma_{\mathsf{K},\infty}^{\frac{1}{2}}\right)\tilde P_\mathsf{K}\left(\Sigma_{\mathsf{K},\infty}^{\frac{1}{2}}A_\mathsf{K}^\top \Sigma_{\mathsf{K},\infty}^{-\frac{1}{2}}\right)
+ I_n
= \tilde P_\mathsf{K}.
$$
By denoting $
\tilde A_\mathsf{K}=
\tilde P_\mathsf{K}^{-\frac{1}{2}}\Sigma_{\mathsf{K},\infty}^{-\frac{1}{2}}A_\mathsf{K}\Sigma_{\mathsf{K},\infty}^{\frac{1}{2}}\tilde P_\mathsf{K}^{\frac{1}{2}}
$, we obtain
$$
\tilde A_\mathsf{K} \tilde A_\mathsf{K}^\top
+\tilde{\rho}(A_{\mK})^2\tilde{P}_{\mK}^{-1}
=\tilde{\rho}(A_{\mK})^2 I_n,
$$
and thus $\|\tilde{A}_\mathsf{K}\|\leq \tilde{\rho}(A_\mathsf{K})$. Consequently,
$$
\left\|\!
\left(\Sigma_{\mathsf{K},\infty}^{-\frac{1}{2}}A_\mathsf{K} \Sigma_{\mathsf{K},\infty}^{\frac{1}{2}} \!\right)^{\!t}
\right\|
=
\left\|
\left(\tilde P_\mathsf{K}^{\frac{1}{2}}\tilde A_K\tilde P_\mathsf{K}^{-\frac{1}{2}}\right)^t
\right\| 
\leq
\left\|\tilde P_\mathsf{K}^{-\frac{1}{2}}\right\|
\left\|\tilde P_\mathsf{K}^{\frac{1}{2}}\right\|
\left\|\tilde A_\mathsf{K}\right\|^t
\!\leq
\varphi(\mathsf{K})\tilde{\rho}(A_\mathsf{K})^{t},
$$
where we define
$
\varphi(\mathsf{K})\coloneqq \left\|\tilde P_\mathsf{K}^{-\frac{1}{2}}\right\|\left\|\tilde P_\mathsf{K}^{\frac{1}{2}}\right\|
$.
It's easy to see that $\varphi(\mathsf{K})\geq \left\|P_\mathsf{K}^{-\frac{1}{2}}P_\mathsf{K}^{\frac{1}{2}}\right\|=1$, and by the results of perturbation analysis of Lyapunov equations \cite{gahinet1990sensitivity}, we can see that $\varphi(\mathsf{K})$ is a continuous function over $\mathsf{K}\in\mathcal{K}_{\mathrm{st}}$.
\end{proof}
\section{Proof of Lemma \ref{lemma:sampling_error}}\label{append: sampling error}

Firstly, notice that $(D_i^0)_{i=1}^N\sim\mathrm{Uni}(\mathbb{S}_{n_K})$ is a direct consequence of the isotropy of the standard Gaussian distribution. The rest of the section will focus on proving \eqref{eq:sampling_error_bound}.

Let $q(t)=[q_1(t) , \cdots, q_N(t)]^\top$, where $q_i(t)$ is defined in {\tt SampleUSphere}. Notice that $q(t)=W q(t-1)$ and $\one^\top q(t)=\one^\top q(t-1)$ for $t\geq 1$. Consequently, 
$$
\frac{Nq(t)}{\one^\top q(0)}-\one
=
\left(W-N^{-1}\one\one^\top\right)
\left(\frac{Nq(t\!-\!1)}{\one^\top q(0)}-\one\right),
$$
where we use the fact that $W$ is doubly stochastic. Thus,
$$
\left|\frac{Nq_i(t)}{\one^{\!\top} q(0)}\!-\!1\right|
\leq
\left\|\frac{Nq(t)}{\one^{\!\top} q(0)}\!-\!\one\right\|
\leq
\rho_W^t \!
\left\|\frac{Nq(0)}{\one^{\!\top} q(0)}\!-\!\one\right\|
\leq
N\rho_W^t,
$$
where the last inequality uses $ \| N v - \one\|\leq N$ for any $v \in \R^N$ such that $v_i\geq 0$ and $\one^\top v=1$.
We then have
\begin{align*}
\left|\frac{1}{\sqrt{Nq_i(T_S)}}
-\frac{1}{\sqrt{\one^\top q(0)}}\right|
=\ &
\frac{1}{\sqrt{\one^\top q(0)}}
\cdot
\frac{\left|Nq_i(T_S)/\one^\top q(0)-1\right|}{
\sqrt{Nq_i(T_S)/\one^\top q(0)}
(1+\sqrt{Nq_i(T_S)/\one^\top q(0)})
} \\
\leq\ &
\frac{1}{\sqrt{\one^\top q(0)}}
\cdot \frac{N\rho_W^{T_S}}{
\sqrt{1 \!-\! N\rho_W^{T_S}}\left(\! 1+\sqrt{1 \!-\!N\rho_W^{T_S}}\right)} \leq \frac{N\rho_W^{T_S}}{\sqrt{\one^\top q(0)}},
\end{align*}
where the last inequality uses that  $N\rho_W^{T_S}\leq 1/2$ and  $\inf_{x\in[0,1/2]}\sqrt{1\!-\!x}\big(1+\sqrt{1\!-\!x}\big)\geq 1$.
Finally, we obtain
\begin{align*}
\|D_i-D_i^0\|_F
=\ &
\|V_i\|_F\left|\frac{1}{\sqrt{Nq_i(T_J)}}
-\frac{1}{\sqrt{1^\top q_0}}\right| \\
\leq\ &
\frac{\|V_i\|_F}{\sqrt{1^\top q_0}}
\cdot N\rho_W^{T_S}
=\|D_i^0\|_F\cdot N\rho_W^{T_S}, \\
\sum_{i=1}^N \|D_i\|_F^2
\leq\ &
\sum_{i=1}^N \left(\|D_i^0\|_F+
\|D_i-D_i^0\|_F\right)^2 \\
\leq\ &
\sum_{i=1}^N \left(1+N\rho_W^{T_S}\right)^2\|D_i^0\|_F^2
=\left(1+N\rho_W^{T_S}\right)^2.
\end{align*}

\section{Proof of Lemma \ref{lem: bias and var of mui}}\label{append: bias and var of cost estimation}

This section analyzes the error of the estimated cost $\tilde J_i(\mK)$, also denoted as $\mu_i(T_J)$,  generated by the subroutine {\tt GlobalCostEst}.

The main insight behind the proof is that $\mu_i(T_J)$ can be represented by quadratic forms of a Gaussian vector (see Lemma \ref{lemma:quad_form_mu_J}). The proof follows by utilizing the properties of the quadratic forms of Gaussian vectors (Proposition \ref{proposition:Gaussian_quadratic_stat}).

\vspace{3pt}
\noindent\textit{(a) Representing $\mu_i(T_J)$ by Quadratic  Gaussian.} We define
\begin{align*}
\varpi\coloneqq &
\begin{bmatrix}
\Sigma_{w}^{-\frac{1}{2}}w(0) \\
\vdots \\
\Sigma_{w}^{-\frac{1}{2}}w(T_J\!-\!1)
\end{bmatrix},
\qquad
\Psi
\!\coloneqq
\!
\begin{bmatrix}
\Sigma_{\mathsf{K},\infty}^{-\frac{1}{2}}
\Sigma_{w}^{\frac{1}{2}}
\\
\Sigma_{\mathsf{K},\infty}^{-\frac{1}{2}}A_{\mathsf{K}}\Sigma_{w}^{\frac{1}{2}} &
\Sigma_{\mathsf{K},\infty}^{-\frac{1}{2}}
\Sigma_{w}^{\frac{1}{2}}
\\
\vdots & \vdots & \ddots
\\
\Sigma_{\mathsf{K},\infty}^{-\frac{1}{2}}A_{\mathsf{K}}^{T_J\!-\!1}
\Sigma_w^{\frac{1}{2}} &
\Sigma_{\mathsf{K},\infty}^{-\frac{1}{2}}A_{\mathsf{K}}^{T_J\!-\!2}\Sigma_{w}^{\frac{1}{2}} &
\cdots &
\Sigma_{\mathsf{K},\infty}^{-\frac{1}{2}}
\Sigma_{w}^{\frac{1}{2}}
\end{bmatrix},\\
\Phi_{\gamma}
\!\coloneqq &\ 
{\Psi}^{\!\top}
\!\!\cdot\!\operatorname{blkdiag}\!
\left[\!
\left(\!\gamma^{T_J-t}\Sigma_{\mathsf{K},\infty}^{\frac{1}{2}}Q_{\mathsf{K}}\Sigma_{\mathsf{K},\infty}^{\frac{1}{2}}\right)_{t=1}^{T_J}
\!\right]
\!\!\cdot\! \Psi,\ \ \gamma\!\in\![0,1],
\end{align*}
where $\operatorname{blkdiag}[(M_l)_{l=1}^p]$ denotes the block diagonal matrix formed by $M_1,\ldots,M_p$. Notice that $\varpi \sim \mathcal{N}(0,I_{nT_J})$.

The following lemma shows that $\mu_i(T_J)$ can be written as a quadratic form in terms of the above auxiliary quantities.

\begin{lemma}[Quadratic Gaussian representation]
\label{lemma:quad_form_mu_J}
For  $\gamma \in[0,1]$, 
 \begin{align}\label{eq:sum_global_cost_quad_form_Phi}
   \sum_{t=1}^{T_J}\gamma^{T_J-t}x(t)^\top Q_{\mathsf{K}}x(t)
=\varpi^\top
\Phi_\gamma\varpi.
\end{align}
Moreover, for any $1\leq i \leq N$, the global objective estimation $\mu_i(T_J)$ (a.k.a. $\tilde J_i(\mK)$) satisfies
\begin{equation}\label{eq:mui_quad_form_Phi}
\left|\mu_i(T_J)\!-\!
\frac{1}{T_J}\varpi^\top
\Phi_1\varpi\right|
\leq
\frac{N}{T_J}\varpi^\top
\Phi_{\rho_W} \varpi. 
\end{equation}
\end{lemma}
\begin{proof} We first prove \eqref{eq:sum_global_cost_quad_form_Phi}.
For a closed-loop system  ${x(t+ 1)}=A_{\mK} x(t)+w(t)$ started with  $x(0)=0$, we have
$$
x(t)=\sum_{\tau=1}^{t} \!A_{\mathsf{K}}^{t-\tau}w(\tau \!-\! 1).
$$
Then, by the definitions of $\varpi$ and $\Phi_{\gamma}$, we have
$$
\begin{aligned}
\sum_{t=1}^{T_J} \gamma^{T_J-t}x(t)^\top Q_{\mK} x(t) 
=\ &
\sum_{t=1}^{T_J}\!\gamma^{T_J-t}
\!\left(\sum_{\tau=1}^{t}w(\tau \!-\! 1)^{\!\top} \!\big(A_{\mK}^\top\big)^{t\!-\tau}\!\right)
\! Q_{\mathsf{K}} \!
\left(\sum_{\tau=1}^{t}A_{\mK}^{t\!-\tau\!} w(\tau \!-\! 1)\!\right) \\
=\ &
\sum_{t=1}^{T_J}(\Psi \varpi)_t^\top \left(\gamma^{T_J-t}\Sigma_{\mathsf{K},\infty}^{\frac{1}{2}}Q_{\mathsf{K}}\Sigma_{\mathsf{K},\infty}^{\frac{1}{2}}\right)(\Psi \varpi)_t=
\varpi^\top
\Phi_{\gamma}\varpi,
\end{aligned}
$$
where $\!(\Psi \varpi)_t\!=\!\sum_{\tau=1}^t\!\Sigma_{\mK,\infty}^{-\frac{1}{2}}A_{\mK}^{t-\tau}w(\tau\!-\!1)$ is the $t$'th block of $\Psi \varpi$.

Next, we prove the inequality \eqref{eq:mui_quad_form_Phi}. Notice that
\begin{align*}
\left|T_J\mu_i(T_J)\!-\!
\varpi^\top
\Phi_1\varpi\right|
=\ & \left|\sum_{t=1}^{T_J}\sum_{j=1}^N \big[W^{T_J-t}\big]_{ij} x(t)^\top Q_{j, \mK} x(t)- \sum_{t=1}^{T_J} x(t)^\top Q_{\mK} x(t)\right|\\
= \ & \left| \sum_{t=1}^{T_J}\sum_{j=1}^N \left(\big[W^{T_J-t}\big]_{ij}-\frac{1}{N}\right) x(t)^\top Q_{j, \mK} x(t)\right|\\
= \ & \left| \sum_{t=1}^{T_J}\sum_{j=1}^N \left[\big(W-\frac{1}{N}\one \one^\top \big)^{T_J-t}\right]_{ij} x(t)^\top Q_{j, \mK} x(t)\right|\\
\leq \ & \sum_{t=1}^{T_J} \left| \sum_{j=1}^N \left[\big(W-\frac{1}{N}\one \one^\top \big)^{T_J-t}\right]_{ij} x(t)^\top Q_{j, \mK} x(t)\right|\\
\leq\ &
\sum_{t=1}^{T_J}
\rho_W^{T_J-t}\sqrt{\sum_{j=1}^N \big(x(t)^\top Q_{j,\mathsf{K}} x(t)\big)^2 }
\leq
\sum_{t=1}^{T_J}
\rho_W^{T_J-t}\sum_{j=1}^N x(t)^\top Q_{j,\mathsf{K}} x(t) \\
=\,&
N \sum_{t=1}^{T_J}\rho_W^{T_J-t} x(t)^\top Q_{\mathsf{K}} x(t)=N\varpi^\top \Phi_{\rho_W} \varpi,
\end{align*}
where the first step uses the definition of $\mu_i(T_J)$ and \eqref{eq:sum_global_cost_quad_form_Phi}; the second step uses the definition of $Q_{\mK}$; the third step uses a property of a doubly stochastic matrix $W$ that $\big(W-\frac{1}{N}\one \one^\top\big)^t= W^t-\frac{1}{N}\one \one^\top$; the fifth step uses the fact that for any vector $v\in\mathbb{R}^N$, we have
$$
\begin{aligned}
  \left|\sum_{j=1}^N \left[\big(W-N^{-1}\one \one^\top \big)^{T_J-t}\right]_{ij}
v_j\right|
\leq\ &
\left\|\big(W-N^{-1}\one\one^\top\big)^{T_J-t}v\right\|_\infty  \\
\leq\ &
\left\|\big(W-N^{-1}\one\one^\top\big)^{T_J-t}v\right\|
\leq \rho_W^{T_J-t}\|v\|;
\end{aligned}
$$
the sixth step follows from $\|v\|\leq \one^\top v$ for any vector $v$ with nonnegative entries; the last step uses \eqref{eq:sum_global_cost_quad_form_Phi}.
\end{proof}

\noindent\textit{(b) Properties of the Parameter Matrices $\Phi_{\rho_W}$ and $\Phi_1$.}
\vspace{-2pt}
\begin{lemma}[Properties of $\Phi_{\rho_W}$ and $\Phi_1$] \label{lemma:Phi_trace_norm_bound}
The parameter matrices $\Phi_{\rho_W}$ and $\Phi_1$ in the quadratic Gaussian representation enjoy the following properties.
\begin{itemize}[leftmargin=10pt]
\item  $\Phi_{\rho_W}$ satisfies
$$
\operatorname{tr}\!\big(
\Phi_{\rho_W}
\big)
\leq \frac{J(\mathsf{K})}{1-\rho_W}.
$$
\item $\Phi_1$ satisfies
$$
\big\|\Phi_1\big\|
\leq
J(\mathsf{K})
\left(
\frac{2\varphi(\mathsf{K})}{1-\rho(A_{\mathsf{K}})}
\right)^2
\qquad\text{and}\qquad
\big\|\Phi_1\big\|_F^2\!
\leq\!
J(\mathsf{K})^2 n T_J\!
\left(
\frac{2\varphi(\mathsf{K})}{1-\rho(A_{\mathsf{K}})}
\right)^4,\!$$
where $\varphi(\mK)$ is defined in Lemma~\ref{lemma:bound_iteratedA}.
\end{itemize}

\end{lemma}
\begin{proof}
{\scriptsize\textbullet}   For $\Phi_{\rho_W}$, by $\varpi\!\sim\! \mathcal N(0, I_{nT_J})$, Lemma~\ref{lemma:quad_form_mu_J}, and $\mathbb{E}\!\left[x(t)x(t)^\top\right] \preceq \Sigma_{\mK, \infty}$, we have
$$
\begin{aligned}
\operatorname{tr}(\Phi_{\rho_W})
=\ &
\mathbb{E}[\varpi^\top
\Phi_{\rho_W}\varpi] =
\mathbb{E}[
\sum_{t=1}^{T_J}
\rho_W^{T_J-t} x(t)^\top Q_{\mK} x(t)
] \\
\leq\ &
\sum_{t=1}^{T_J}
\rho_W^{T_J-t} \operatorname{tr}(Q_{\mK}
\Sigma_{\mK,\infty}
)
=
J(\mathsf{K}) \sum_{t=1}^{T_J}
\rho_W^{T_J-t}
\leq \frac{J(\mathsf{K})}{1-\rho_W}.
\end{aligned}
$$

\noindent{ \scriptsize\textbullet}   For $\Phi_1$, since the matrix $\Sigma_{\mathsf{K},\infty}^{\frac{1}{2}} Q_{\mathsf{K}} \Sigma_{\mathsf{K},\infty}^{\frac{1}{2}}$ is positive definite, we have that
$$
\begin{aligned}
\big\|\Phi_1\big\|
\leq
\left\|\Sigma_{\mathsf{K},\infty}^{\frac{1}{2}} Q_{\mathsf{K}} \Sigma_{\mathsf{K},\infty}^{\frac{1}{2}}\right\|\big\|\Psi\big\|^2
\leq\ &
\operatorname{tr}(\Sigma_{\mathsf{K},\infty}^{\frac{1}{2}} Q_{\mathsf{K}} \Sigma_{\mathsf{K},\infty}^{\frac{1}{2}})
\big\|\Psi\big\|^2 \\
=\ &
\operatorname{tr}(Q_{\mK}\Sigma_{\mK,\infty})\big\|\Psi\big\|^2
= 
J(\mathsf{K})
\big\|\Psi\big\|^2.
\end{aligned}
$$

Now it suffices to bound $\big\|\Psi\big\|^2.$ Consider any $v=\begin{bmatrix}v_0^\top & v_1^\top & \cdots & v_{T_J-1}^\top\end{bmatrix}^\top\in\mathbb{S}_{nT_J}$, and then we have
\begin{align*}
\big\|\Psi v\big\|^2
 =\, &\sum_{t=1}^{T_J}
\sum_{\tau,\tau'=1}^t
\!\! v_{\tau -\! 1}^\top\Sigma_w^{\frac{1}{2}} \!
\left(
A_{\mathsf{K}}^\top
\right)^{t - \tau} 
\Sigma_{\mathsf{K},\infty}^{-1}
A_{\mathsf{K}}^{t-\tau'}
\Sigma_w^{\frac{1}{2}}
v_{\tau' \!-\! 1} \\
 =\, &
\sum_{t=1}^{T_J}
\sum_{\tau,\tau'=1}^t
\!\!v_{\tau -\! 1}^\top\Sigma_w^{\frac{1}{2}}\Sigma_{\mathsf{K},\infty}^{-\frac{1}{2}}
\big(
\hat A_{\mathsf{K}}^\top
\big)^{\!t\!-\!\tau} 
\hat A_{\mathsf{K}}^{t-\tau'}
\!\!\Sigma^{-\frac{1}{2}}_{\mK,\infty}\Sigma_w^{\frac{1}{2}}
v_{\tau' \!-\! 1} \\
 \leq \, &
\sum_{t=1}^{T_J}
\sum_{\tau,\tau'=1}^t
\!\|v_{\tau-1}\|
\left\|\big(\hat A_{\mathsf{K}}^\top
\big)^{t-\tau}\right\| 
\left\|\big(\hat A_{\mathsf{K}}\big)^{t-\tau'}\right\|\|v_{\tau'-1}\| \\
\leq \, &
\varphi(\mathsf{K})^2
\sum_{t=1}^{T_J}
\sum_{\tau,\tau'=1}^t
\!\!\|v_{\tau \!-\! 1}\|
\!\cdot\tilde{\rho}(A_{\mathsf{K}})^{t-\tau}
\!\!\cdot
\tilde{\rho}(A_{\mathsf{K}})^{t-\tau'}
\!\cdot\! \|v_{\tau' \!-\! 1}\|\\
= \ & \! \varphi(\mathsf{K})^2 \! \left\|
\underbrace{\begin{bmatrix}
	1 \\
	\tilde{\rho}(A_{\mathsf{K}}) & 1 \\
	\vdots & \vdots & \ddots \\
	\tilde{\rho}(A_{\mathsf{K}})^{T_J-1} &
	\tilde{\rho}(A_{\mathsf{K}})^{T_J-2} &
	\cdots & 1
	\end{bmatrix}}_{\mathbf H^{(T_J)}} 
	\underbrace{
	\begin{bmatrix}
	\|v_0\| \\ \|v_1\| \\ \vdots \\ \|v_{T-1}\|
	\end{bmatrix}}_{\mathbf v}
\right\|^2\\
\leq  \ &\varphi(\mathsf{K})^2 \big\|\mathbf H^{(T_J)}\big\|^2\|\mathbf v\|^2= \varphi(\mathsf{K})^2\|\mathbf H^{(T_J)}\|^2
\end{align*}
where in the second step we denote $\hat A_{\mK}=\Sigma_{\mathsf{K},\infty}^{-\frac{1}{2}}A_{\mathsf{K}}
\Sigma_{\mathsf{K},\infty}^{\frac{1}{2}}$; the third step uses $\Sigma_w\preceq\Sigma_{\mathsf{K},\infty}$ and thus $\big\|\Sigma_{\mathsf{K},\infty}^{-1/2}\Sigma_{w}^{1/2}\big\|\leq 1$; the fourth step uses  Lemma~\ref{lemma:bound_iteratedA} and $\tilde\rho(A_{\mathsf{K}})=(1+\rho(A_{\mathsf{K}}))/2$; the last step is because $v$ is on the unit sphere.

Notice that $\mathbf H^{(T_J)}$ can be viewed as a finite-horizon truncation of the block-Toeplitz representation of  the linear system with transfer function 
$$
\mathbf{H}(z)=
\sum_{t=0}^\infty
\tilde{\rho}(A_{\mK})^t z^{-t}
=\frac{1}{1\!-\!\tilde{\rho}(A_{\mK})z^{-1}}.$$ Therefore, 
$$
\big\|\mathbf H^{(T_J)}\big\|
\!\leq 
\|\mathbf{H}\|_{\mathcal{H}_\infty}\!=\!
\sup_{\|z\|=1}\!|\mathbf{H}(z)|
\!=\!\frac{1}{1 \!-\! \tilde{\rho}(A_{\mK})}\!=\!\frac{2}{1 \!-\! {\rho}(A_{\mK})}.
$$
This completes the proof of the bound on $\|\Phi_1\|$. The  bound on $\big\|\Phi_1\big\|_F^2$
follows from $\big\|\Phi_1\big\|_F^2
\leq nT_J \big\|\Phi_1\big\|^2$.
\end{proof}

\noindent\textit{(c) Bounding the Bias and Second Moment.} The proof relies on the following properties of quadratic Gaussian variables.
\vspace{-2pt}
\begin{proposition}[{\hspace{1sp}\cite[Theorems 1.5 \& 1.6]{seber2003linear}}]\label{proposition:Gaussian_quadratic_stat}
Let $z\sim\mathcal{N}(0, I_p)$, and let $M\in\mathbb{R}^{p\times p}$ be any symmetric matrix. Then we have $\mathbb{E}\!\left[z^\top M z\right]=\operatorname{tr}(M)$ and $\operatorname{Var}\!\left(z^\top M z\right)=2 \big\|M\|_F^2$.
\end{proposition}

Now we are ready for the proof of Lemma \ref{lem: bias and var of mui}.

Firstly, we consider the bias. 
\begin{align*}
\left|\E[\mu_i(T_J)]- J(\mK)\right|
\leq &\left| \E\!\left[\mu_i(T_J)- \frac{1}{T_J}\varpi^\top
\Phi_1\varpi\right]\right|+\left|\E\!\left[\frac{1}{T_J}\varpi^\top
\Phi_1\varpi\right]- J(\mK)\right|\\
 \leq &\  \E\!\left[\frac{N}{T_J}\varpi^\top
\Phi_{\rho_W} \varpi\right] + \left| \frac{1}{T_J}\E\!\left[\sum_{t=1}^{T_J}x(t)^\top Q_{\mK} x(t)\right]-J(\mK)\right|\\
 =& \ \frac{N}{T_J}\operatorname{tr}(\Phi_{\rho_W})+\frac{1}{T_J}\left|\sum_{t=1}^{T_J}\operatorname{tr}(Q_{\mK}(\Sigma_{\mK, t}-\Sigma_{\mK,\infty}))\right|\\
 \leq &\ \frac{N}{T_J}\frac{J(\mK)}{1-\rho_W}
 +\frac{1}{T_J}\sum_{t=1}^{T_J}\left| \operatorname{tr}\!\bigg(Q_{\mK} \sum_{\tau=t}^{\infty} A_{\mK}^\tau \Sigma_w \big(A_{\mK}^\top\big)^\tau\bigg)\right|\\
 \leq &\ \frac{N}{T_J}\frac{J(\mK)}{1-\rho_W}+\frac{J(K)}{T_J}\left(\frac{2\varphi(\mK)}{1-\rho(A_{\mK})}\right)^2
\end{align*}
where the second step uses Lemma~\ref{lemma:quad_form_mu_J}; the third step follows from $\varpi\sim \mathcal N(0, I_{nT_J})$, \eqref{equ: Sigma(K,t) def} and \eqref{equ: formula of J(K)}; the fourth step uses Lemma~\ref{lemma:Phi_trace_norm_bound} and \eqref{equ: Sigma(K,t)}; the last step uses the following fact:
\begin{align*}
\sum_{t=1}^{T_J}\left| \operatorname{tr}\!
\bigg(Q_{\mK} \sum_{\tau=t}^{\infty} A_{\mK}^\tau \Sigma_w \big(A_{\mK}^\top\big)^\tau\bigg)\right|
= \,&\sum_{t=1}^{T_J}\left| \operatorname{tr}\!\bigg(\Sigma_{\mK, \infty}^{\frac{1}{2}}Q_{\mK} \Sigma_{\mK, \infty}^{\frac{1}{2}}\sum_{\tau=t}^{\infty} \hat A_{\mK}^\tau \hat \Sigma_w \big(\hat A_{\mK}^\top\big)^\tau\bigg)\right|\\
\leq \,& \sum_{t=1}^{T_J} \operatorname{tr}\!\big(\Sigma_{\mK, \infty}^{\frac{1}{2}}Q_{\mK} \Sigma_{\mK, \infty}^{\frac{1}{2}}\big) \left\| \sum_{\tau=t}^{\infty} \hat A_{\mK}^\tau \hat \Sigma_w \big(\hat A_{\mK}^\top\big)^\tau\right\|\\
\leq \,& \sum_{t=1}^{T_J}J(\mK) \sum_{\tau=t}^{\infty} \big\|  \hat A_{\mK}^\tau\big\|^2 \big\|\hat \Sigma_w\big\|\\
\leq \,& \sum_{t=1}^{T_J}\! J(\mK) \sum_{\tau=t}^{\infty} \!\varphi(\mK)^2\!
\left(\frac{1 \!+\! \rho(A_{\mK})}{2}\right)^{\!\!2\tau}
\!\!\leq J(\mK)  \varphi(\mK)^2
\!\left(\frac{2}{1\!-\!\rho(A_{\mK})}\!\right)^{\!2} \!\!,
\end{align*}
where we denote $\hat A_{\mK}=\Sigma_{\mK, \infty}^{-\frac{1}{2}}A_{\mK}\Sigma_{\mK, \infty}^{\frac{1}{2}}$ and $\hat \Sigma_w=\Sigma_{\mK, \infty}^{-\frac{1}{2}}\Sigma_w \Sigma_{\mK, \infty}^{-\frac{1}{2}}$, the fourth step uses Lemma~\ref{lemma:bound_iteratedA} and $\big\|\hat \Sigma_w\big\|\leq 1$,  the last step uses $\sum_{t=1}^\infty\sum_{\tau=t}^\infty\left(\mfrac{1 \!+\! \rho(A_{\mK})}{2}\right)^{2\tau}
\leq \left(\mfrac{2}{1 \!-\! \rho(A_{\mK})}\right)^2$.

Define the constant $\beta_0$ as the following:
\begin{equation}\label{equ: beta0 def}
\beta_0\coloneqq \sup_{\mathsf{K}\in \mathcal Q^1}
\left(\frac{2\varphi(\mathsf{K})}{1-\rho(A_{\mathsf{K}})}\right)^2.
\end{equation}
Lemmas \ref{lemma:J_smoothness}, \ref{lemma:bound_iteratedA} and the continuity of the map $\mK\mapsto\rho(A_{\mK})$ ensure that $\beta_0$ is finite and only depends on the system parameters $A$, $B$, $\Sigma_w$,$Q_i$, $R_i$ as well as the initial cost $J(\mathsf{K}_0)$. By substituting $\beta_0$ into the inequality above, we prove  \eqref{eq:bias_mu_i} for any $\mK\in \Q^1$. 

Next, we bound $\mathbb{E}
\!\left[(\mu_i(T_J)-J(\mathsf{K}))^2\right]$. By \eqref{eq:bias_mu_i}, we have
$$
\begin{aligned}
\E\!\left[(\mu_i(T_J)\!-\!J(\mathsf{K}))^2\right] 
\leq\ &
(\mathbb{E}
\!\left[\mu_i(T_J)-J(\mathsf{K})\right])^2+ \operatorname{Var}(\mu_i(T_J)) \\
\leq\ &
\frac{J(\mathsf{K})^2\!}{T_J^2}\!
\left[
2\left(\!\frac{N}{1\!-\!\rho_W}\!\right)^2\!
\!+\!2\beta_0^2
\right]\!+\! \operatorname{Var}(\mu_i(T_J)).
\end{aligned}
$$
Then, we can  
bound then $\operatorname{Var}(\mu_i(T_J))$ below:
\begin{align*}
\operatorname{Var}(\mu_i(T_J))
\leq\ &
2\operatorname{Var}\!\left(\mu_i(T_J)\!-\!
\frac{1}{T_J}\varpi^\top\!
\Phi_1\varpi\right)  +
2\operatorname{Var}\!\left(
\frac{1}{T_J}\varpi^\top\!
\Phi_1\varpi\right) \\
\leq\ &
2\,\mathbb{E}\!\left[
\left|\mu_i(T_J)-
\frac{1}{T_J}\varpi^\top
\Phi_1\varpi\right|^2
\right] +
\frac{4}{T_J^2}\left\|\Phi_1\right\|_F^2\\
\leq\ &
2\,\frac{N^2}{T_J^2}\,\mathbb{E}\!\left[
\left(\varpi^\top
\Phi_{\rho_W}\varpi\right)^2
\right]+\frac{4nJ(\mathsf{K})^2}{T_J}
\beta_0^2\\
=\ &
\frac{2N^2}{T_J^2}\!\left(\!\left(\mathbb{E}
[\varpi^\top\!
\Phi_{\rho_W}\!\varpi] \right)^2 \!+\! \operatorname{Var}\!\left(\varpi^\top\!
\Phi_{\rho_W}\!\varpi\right)\!\right)
+\frac{4nJ(\mathsf{K})^2\!}{T_J}\!
\beta_0^2\\
\leq\ &
\frac{6J(\mathsf{K})^2}{T_J^2}\left(\frac{N}{1-\rho_W}\right)^2
+
\frac{4nJ(\mathsf{K})^2}{T_J}
\beta_0^2.
\end{align*}
where we use Proposition~\ref{proposition:Gaussian_quadratic_stat}, Lemmas~\ref{lemma:quad_form_mu_J} and \ref{lemma:Phi_trace_norm_bound},  $\|M\|_F \leq \operatorname{tr}(M)$ for any positive semidefinite $M$, $\mK \in \Q^1$ and \eqref{equ: beta0 def}. Finally, we obtain \eqref{eq:square_diff_mu_i_J} by $1/T_J^2 \leq n/T_J$.

\section{Proof of Lemma \ref{lemma:capped_J_bias}}\label{append: truncation error}
The proof is based on the following concentration inequality.

\begin{proposition}[\hspace{1sp}\cite{hsu2012tail}]\label{prop: tail_gaussian_general}
Let $z\sim \mathcal{N}(0, I_p)$, and let $M\in\mathbb{R}^{p\times p}$ be any symmetric positive definite matrix. Then for any $\delta\geq 0$,
$$
\mathbb{P}\!\left(
z^\top M z
>\operatorname{tr} M
+2\|M\|_F\sqrt{\delta}
+2\|M\|\delta
\right)
\leq e^{-\delta}.
$$

\end{proposition}

By Lemma \ref{lemma:quad_form_mu_J}, we have
$\mu_i(T_J)
\leq T_J^{-1}\varpi^\top
\!\left(\Phi_1 \!+\! N\Phi_{\rho_W}\right)\varpi$.
Therefore for any $\varepsilon_1\geq  1$ and $\varepsilon_2\geq 0$, we have 
$$
\begin{aligned}
\mathbb{P}\left(
\mu_i(T_J)
> (\varepsilon_1+\varepsilon_2)J(\mathsf{K})
\right)
\leq\ &
\mathbb{P}\!\left(\!
\frac{1}{T_J}\varpi^{\!\top}
\!\!\left(
\Phi_1
\!+\!
N \Phi_{\rho_W}
\right)\!
\varpi
> (\varepsilon_1 \!+\! \varepsilon_2)J(\mathsf{K})
\!\right) \\
\leq\ &
\mathbb{P}\!\left(\!
\frac{1}{T_J}\varpi^{\!\top}
\Phi_1\varpi
>
\varepsilon_1J(\mathsf{K})
+
\frac{\operatorname{tr}\!\big(\Phi_1\big)}{T_J}
-J(\mathsf{K})
\!\right) \\
&+
\mathbb{P}\!\left(\!
\frac{N}{T_J}\varpi^{\!\top}
\!\Phi_{\rho_W}\varpi
>
\varepsilon_2 J(\mathsf{K})
+
J(\mathsf{K})
-
\frac{\operatorname{tr}\!\big(\Phi_1\big)}{T_J}
\!\right) \\
\leq\ &
\mathbb{P}\!\left(\!
\frac{1}{T_J}\varpi^{\!\top}
\Phi_1\varpi
>
\varepsilon_1J(\mathsf{K})
+
\frac{\operatorname{tr}\!\big(\Phi_1\big)}{T_J}
-J(\mathsf{K})
\!\right) \\
& +
\mathbb{P} \!\left(\!
\frac{N}{T_J}\varpi^\top
\Phi_{\rho_W}\varpi
>\varepsilon_2J(\mathsf{K})\right),
\end{aligned}
$$
where we used
$$
J(\mathsf{K})\geq
T_J^{-1}\sum_{t=1}^{T_J}\mathbb{E}\!\left[x(t)^\top Q_{\mK}x(t)\right]
=T_J^{-1}\operatorname{tr}\!\big(\Phi_1\big)
$$ by \eqref{equ: Sigma(K,t)}, \eqref{eq:sum_global_cost_quad_form_Phi} and Proposition~\ref{proposition:Gaussian_quadratic_stat}.
For the first term, by Proposition \ref{prop: tail_gaussian_general} and the bound $\big\|\Phi_1\big\|_F\leq \sqrt{nT_J}\big\|\Phi_1\big\|$, we get
$$
\begin{aligned}
\mathbb{P}
\Big(\varpi^{\!\top}
\Phi_1\varpi
>\operatorname{tr}\!\big(\Phi_1\big)
+2\big\|\Phi_1\big\|\sqrt{nT_J\delta}
+ 2\big\|\Phi_1\big\|\delta
\Big)
\leq
e^{-\delta},
\end{aligned}
$$
for any $\delta\geq 0$, and by letting $\delta$ satisfy
$$
2\big\|\Phi_1\big\|\sqrt{nT_J\delta}
+ 2\big\|\Phi_1\big\|\delta
=(\varepsilon_1-1)T_J J(\mathsf{K})$$
with $\varepsilon_1\geq 1$, we can get
\begin{align*}
\mathbb{P}
\!\left(\!
\frac{1}{T_J}\varpi^{\!\top}
\Phi_1\varpi
\!>\!
\varepsilon_1 J(\mathsf{K})
\!+\!
\frac{\operatorname{tr}\!\big(\Phi_1\big)}{T_J}
\!-\!
J(\mathsf{K})
\!\right)
\leq\,&
\exp \!\left[
-\frac{1}{4}
\left(\sqrt{2 \frac{(\varepsilon_1-1)T_J J(\mathsf{K})}{\big\|\Phi_1\big\|}+nT_J}-\sqrt{nT_J}\right)^2\right] \\
\leq\ &
\exp \!\left[
-\frac{1}{4}
\min\!\left\{
\frac{(\varepsilon_1-1)T_J J(\mathsf{K})}{\big\|\Phi_1\big\|},
\frac{(\varepsilon_1-1)^2T_J J(\mathsf{K})^2}{4n\big\|\Phi_1\big\|^2}
\right\}\right] \\
\leq\,&
\exp \! \left(
\!-
\frac{(\varepsilon_1-1)T_J J(\mathsf{K})}{4\big\|\Phi_1\big\|}\right)
+
\exp\left(
\!-
\frac{(\varepsilon_1-1)^2T_J J(\mathsf{K})^2}{16n\big\|\Phi_1\big\|^2}\right)
\!,
\end{align*}
where we used
$
\left(\sqrt{2x \!+\! nT_J} \!-\! \sqrt{nT_J}\right)^{\!2}
\geq\min\!\left\{x,x^2/(4nT_J)\right\}
$ for all $x\geq 0$
in the second inequality. For the second term, by Proposition \ref{prop: tail_gaussian_general} and the bound $\big\|\Phi_{\rho_W}\big\|\leq \big\|\Phi_{\rho_W}\big\|_F\leq \operatorname{tr}\big(\Phi_{\rho_W}\big)$, we obtain
$$
\mathbb{P}
\!\left(\varpi^{\!\top}
\Phi_{\rho_W}\varpi
>\operatorname{tr}\!\big(\Phi_{\rho_W}\big)
\!\left(1
\!+\!
2\sqrt{\delta}
\!+\!
2\delta\right)
\right)
\leq e^{-\delta}
$$
for any $\delta\geq 0$, and by letting
$$
\delta=\frac{1}{4}\!\left(\sqrt{\frac{2T_J\varepsilon_2J(\mathsf{K})}{N\operatorname{tr}\big(\Phi_{\rho_W}\big)}-1}-1\right)^2$$
for $\varepsilon_2\geq N\operatorname{tr}\big(\Phi_{\rho_W}\big)/(T_J J(\mathsf{K}))$, we obtain
$$
\begin{aligned}
\mathbb{P}
\!\left(\!\frac{N\varpi^{\!\top}
\!\Phi_{\rho_W}\varpi}{T_J}
\!>\!
\varepsilon_2 J(\mathsf{K})
\!\right)
\!\leq&
\exp \! \left[\!
-\frac{1}{4}\!\left(\!\!\sqrt{2\frac{\varepsilon_2T_JJ(\mathsf{K})}{N\operatorname{tr}\!\big(\Phi_{\rho_W}\!\big)} \!-\! 1} \!-\!1 \!\right)^{\!\!2}
\right] \\
\!\leq \!& 
\exp \! \left[\!
-\frac{1}{3}
\!\left(\frac{\varepsilon_2T_JJ(\mathsf{K})}{N\operatorname{tr}\!\big(\Phi_{\rho_W} \!\big)}\!-\!2\right)
\!\right],
\end{aligned}
$$
where we used
$\left(\sqrt{2x \!-\! 1} \!-\! 1\right)^{\!2}
\geq \frac{4}{3}(x-2)$ for any $x > 1$.

Thus, by letting $\varepsilon_1=4\varepsilon/5$ and $\varepsilon_2=\varepsilon/5$, we obtain
$$
\begin{aligned}
\mathbb{P}\left(
\mu_i(T_J)
> \varepsilon J(\mathsf{K})
\right)
\leq\,&
\exp\left(
-
\frac{(4\varepsilon/5-1)T_J J(\mathsf{K})}{4\big\|\Phi_1\big\|}\right)
+
\exp\left(
-
\frac{(4\varepsilon/5-1)^2T_J J(\mathsf{K})^2}{16n\big\|\Phi_1\big\|^2}\right)
\\
&+
\exp\left[
-\frac{1}{3}
\!\left(\!\frac{\varepsilon T_JJ(\mathsf{K})}{5 N\operatorname{tr}\big(\Phi_{\rho_W}\big)}-2\right)
\right]
\end{aligned}
$$
for $\varepsilon\geq 5N\operatorname{tr}\big(\Phi_{\rho_W}\big)/(T_JJ(\mathsf{K}))$. Now we have
$$
\begin{aligned}
\mathbb{E}\left[
\mu_i(T_J)-\min\!\left\{\mu_i(T_J),\bar J\right\}\right]
=\,&
\int_0^{+\infty}
\mathbb{P}\!\left(
\mu_i(T_J)-\min\!\left\{\mu_i(T_J),\bar J\right\}
\geq x
\right)\,dx \\
=\,&
\int_0^{+\infty}
\mathbb{P}\!\left(
\mu_i(T_J)
\geq \bar{J}+x
\right)\,dx \\
=\,&
J(\mathsf{K})\int_{\bar{J}/J(\mathsf{K})}^{+\infty}
\mathbb{P}\!\left(
\mu_i(T_J)
\geq \varepsilon J(\mathsf{K})
\right)\,d\varepsilon.
\end{aligned}
$$
By using
$e^{-x}<1/(2x)$ and 
$\int_x^{+\infty}\! e^{-u^2}du<{e^{-x^2}}/{(2x)}$ for any $x>0$, we can see that
\begin{align*}
\int_{\bar{J}/J(\mathsf{K})}^{+\infty}
\exp\left(
-
\frac{(4\varepsilon/5-1)T_J J(\mathsf{K})}{4\big\|\Phi_1\big\|}\right)
\,d\varepsilon
=\ &
\frac{5\big\|\Phi_1\big\|}{T_J J(\mathsf{K})}
\exp\left[
-\frac{T_J J(\mathsf{K})}{4\big\|\Phi_1\big\|}
\left(\frac{4 \bar J}{5J(\mathsf{K})}-1\right)
\right] \\
<\ &
\frac{10\big\|\Phi_1\big\|^2}{T_J^2 J(\mathsf{K})^2}
\cdot
\frac{1}{4\bar{J}/(5J(\mathsf{K}))-1},\\
\int_{\bar{J}/J(\mathsf{K})}^{+\infty}
\exp\left(
-
\frac{(4\varepsilon/5-1)^2T_J J(\mathsf{K})^2}{16n\big\|\Phi_1\big\|^2}\right)
\,d\varepsilon
<\ &
\frac{10 n \big\|\Phi_1\big\|^2}{T_J J(\mathsf{K})^2
}
\frac{\exp\left[
-
\frac{T_J J(\mathsf{K})^2}{16n\big\|\Phi_1\big\|^2}
\left(\frac{4\bar J}{5J(\mathsf{K})}-1\right)^2\right]}
{\frac{4\bar J}{5J(\mathsf{K})}-1}\\
<\ &
\frac{80 n^2 \big\|\Phi_1\big\|^4}{
T_J^2 J(\mathsf{K})^4}
\frac{1}{\left(4\bar J/(5J(\mathsf{K}))-1\right)^3}, \\
\int_{\bar{J}/J(\mathsf{K})}^{+\infty}
\exp\!\left[
-\frac{1}{3}
\!\left(\!
\frac{\varepsilon T_JJ(\mathsf{K})}{5N\operatorname{tr}\!\big(\Phi_{\rho_W}\big)}\!-\!2\right)
\!\right]
\,d\varepsilon
=\ &
\frac{15e^{\frac{2}{3}} N\operatorname{tr}\!\big(\Phi_{\rho_W}\big)}{T_J J(\mathsf{K})}
\exp \!\left(\!
-
\frac{\bar{J} T_J}{15N\operatorname{tr}\!\big(\Phi_{\rho_W}\big)}
\right) \\
<\ &
\frac{225 N^2\operatorname{tr}\!\big(\Phi_{\rho_W}\big)^2}{T_J^2 J(\mathsf{K})^2}
\cdot\frac{J(\mathsf{K})}{\bar{J}}.
\end{align*}
Finally, by Lemma \ref{lemma:Phi_trace_norm_bound} and the 
condition on $\bar{J}$, we see that
\begin{align*}
\mathbb{E}\left[
\mu_i(T_J)-\min\!\left\{\mu_i(T_J),\bar J\right\}\right]
\leq\,&
\frac{J(\mathsf{K})}{T_J^2}
\bigg[10
\!\left(\!\frac{2\varphi(\mathsf{K})}{1 \!-\!\rho(A_{\mathsf{K}})}\!\right)^{\!\!4}
\!+\!
80n^2
\!\left(\!\frac{2\varphi(\mathsf{K})}{1\!-\!\rho(A_{\mathsf{K}})}\!\right)^{\!\!8}
\!+\!
\frac{90 N^2}{(1\!-\!\rho_W)^2}
\bigg] \\
\leq\ &
\frac{90 J(\mathsf{K})}{T_J^2}
\left[
n^2\left(\frac{2\varphi(\mathsf{K})}{1-\rho(A_{\mathsf{K}})}\right)^8
+\frac{N^2}{(1-\rho_W)^2}
\right]\\
\leq\ &
\frac{90 J(\mathsf{K})}{T_J^2}
\left[
n^2\beta_0^4
+\frac{N^2}{(1-\rho_W)^2}
\right].
\end{align*}
The inequality $\mathbb{E}\left[
\mu_i(T_J)-\min\!\left\{\mu_i(T_J),\bar J\right\}\right]\geq 0$ is obvious.

\bibliographystyle{IEEEtran.bst}
\bibliography{references.bib}
\end{document}